\documentclass[twoside]{article}

\usepackage[accepted]{aistats2021}
\usepackage[numbers]{natbib}

\usepackage{amssymb}
\usepackage{amsmath, bbm}
\usepackage{amsthm}
\usepackage{xr-hyper} 
\usepackage{hyperref}
\usepackage{graphicx}

\usepackage[dvipsnames]{xcolor}
\usepackage{tikz}
\usetikzlibrary{arrows.meta}
\usetikzlibrary{decorations.pathreplacing}

\usepackage{algorithm}
\usepackage{algpseudocode}
\usepackage[dvipsnames]{xcolor}
\algnewcommand\algorithmicinput{\textbf{INPUT:}}
\algnewcommand\INPUT{\item[\algorithmicinput]}
\algnewcommand\algorithmicoutput{\textbf{OUTPUT:}}
\algnewcommand\OUTPUT{\item[\algorithmicoutput]}

\DeclareMathAlphabet{\mathpzc}{OT1}{pzc}{m}{it}

\usepackage{cleveref}

\newtheorem{theorem}{Theorem}
\newtheorem{lemma}[theorem]{Lemma}
\newtheorem{proposition}[theorem]{Proposition}
\newtheorem{corollary}[theorem]{Corollary}

\newtheorem{assume}{Assumption}
\newtheorem{model}{Model}

\allowdisplaybreaks

\usepackage[textsize=tiny]{todonotes}
\DeclareMathOperator*{\argmin}{arg\,min}

\newcommand{\hatp}{ \widehat {\mathcal P}} 

\begin{document}

%

%

\twocolumn[

\aistatstitle{Localizing Changes in High-Dimensional Regression Models}

\aistatsauthor{ Alessandro Rinaldo \And  Daren Wang \And Qin Wen  }

\aistatsaddress{ Department of Statistics  \& Data Science 
\\Carnegie Mellon University 
\And Department of Statistics 
\\ University of Chicago \And Department of Statistics
\\ University of Chicago  } 
\aistatsauthor{   Rebecca Willett \And Yi Yu }
\aistatsaddress{    Department of Statistics 
\\  University of Chicago   \And  Department of Statistics 
\\  University of Warwick }
]

\begin{abstract}
This paper addresses the problem of localizing change points in high-dimensional linear regression models with piecewise constant regression coefficients. We develop a dynamic programming approach to estimate the locations of the change points whose performance improves upon the current state-of-the-art, even as the dimensionality, the sparsity of the regression coefficients, the temporal spacing between two consecutive change points, and the magnitude of the difference of two consecutive regression coefficient vectors are allowed to vary with the sample size. Furthermore, we devise a computationally-efficient refinement procedure that provably reduces the localization error of preliminary estimates of the change points. We demonstrate minimax lower bounds on the localization error that nearly match the upper bound on the localization error of our methodology and show that the signal-to-noise condition we impose is essentially the weakest possible based on information-theoretic arguments. Extensive numerical results support our theoretical findings, and experiments on real air quality data reveal change points supported by historical information not used by the algorithm.
\end{abstract}

\section{INTRODUCTION}\label{sec-intro}

High-dimensional linear regression  modeling has been extensively applied and studied over the last two decades due to the technological advancements in collecting and storing data from a wide range of application areas, including biology, neuroscience, climatology, finance, cybersecurity, to name but a few. There exist now a host of methodologies available to practitioners to fit high-dimensional sparse linear models, and their properties have been thoroughly investigated and are now well understood. See \cite{buhlmann2011statistics} for recent reviews. 

In this paper, we are concerned with a non-stationary variant of the high-dim linear regression model in which the data are observed as a time series and the regression coefficients are piece-wise stationary, with changes occurring at unknown times. We formally introduce our model settings next.

\begin{model}\label{assume:change point regression model}
	Let the data  $\{(x_t, y_t)\}_{t = 1}^n \subset \mathbb{R}^p \times \mathbb{R}$ satisfy the model
		\[
			y_t = x_t^\top \beta_t^* + \varepsilon_t, \quad t =1, \ldots, n
		\]
		where $\{\beta_t^*\}_{t = 1}^n \subset \mathbb{R}^p$ is the unknown coefficient vector, $\{x_t\}_{t = 1}^n$ are independent and identically distributed mean-zero sub-Gaussian random  vectors with $\mathbb{E}(x_t x_t^{\top}) = \Sigma$,  and $\{\varepsilon_t\}_{t = 1}^n$ are independent mean-zero sub-Gaussian random variables with sub-Gaussian parameter bounded by $\sigma_{\varepsilon}^2$ and independent of $\{x_t\}_{t=1}^n$. 
	In addition, there exists a sequence of change points $1 = \eta_0 < \eta_1 < 
	\ldots < \eta_{K+1} = n$ such that $\beta_t^* \neq \beta_{t - 1}^*$, if and only if $t \in \{\eta_k\}_{k = 1}^K$.
\end{model}

We consider a high-dimensional framework where the features of the above change-point model  are allowed to change with the sample size $n$; see \Cref{assume:high dim coefficient} below for details. Given data sampled from \Cref{assume:change point regression model}, our main task is to develop computationally-efficient algorithms that can consistently estimate both the unknown number $K$ of change points and the time points $\{\eta_k\}_{k = 1}^K$, at which the regression coefficients change.  That is, we seek \emph{consistent} estimators $\{\hat{\eta}_k\}_{k = 1}^{\widehat{K}}$, such that, as the sample size $n$ grows unbounded, it holds with probability tending to 1 that 
	\[
		\widehat{K} = K \quad \mbox{and} \quad \epsilon = \max_{k = 1, \ldots, K} |\widehat \eta_k -\eta_k|  = o(\Delta),
	\]
	where $\Delta = \min_{k=1,\ldots,K_1} \eta_k - \eta_{k-1}$ is the minimal spacing between consecutive change-points.
	We refer to the quantity $\epsilon$ as the {\it localization error rate}.

The model detailed above has already been considered in the recent literature.  \cite{lee2016lasso}, \cite{kaul2018parameter}, \cite{lee2018oracle}, among others, focused on the cases where there exists at most one true change point.  \cite{leonardi2016computationally} and \cite{zhang2015change} considered multiple change points and devised consistent change point estimators, albeit with localization error rates worse than the one we establish in \Cref{eq:DP consistent for regression}.  \cite{wang2019statistically} also allowed for multiple change points in a regression setting and proposed a variant of the wild binary segmentation  (WBS) method \citep{fryzlewicz2014wild}, the performances thereof match the one of the procedure we study next.   More detailed comparisons are further commentary can be found in \Cref{sec-comparisons}.

In this paper, we make several theoretical and methodological contributions, summarized next, that improve the existing literature. 	

\begin{itemize}
\item We provide consistent change point estimators for \Cref{assume:change point regression model}.  We allow for model parameters to change with the sample size $n$, including the dimensionality of the data, the entry-wise sparsity of the coefficient vectors, the number of change points, the smallest distance between two consecutive change points, and the smallest difference between two consecutive different regression coefficients.   To the best of our knowledge, the theoretical results we provide in this paper are the sharpest in the existing literature.  Furthermore, the proposed algorithms, based on the general framework described  in \eqref{eq-wide-p}, can be implemented using dynamic programming approaches and are computationally efficient.
\item We devise a additional second step (\Cref{algorithm:LR}), called local refinement, that is guaranteed to deliver an even better localization error rate, even though directly optimizing \eqref{eq-wide-p} already provides the sharpest rates among the ones existing in the literature.
\item We present information-theoretic lower bounds on both detection and localization, establishing the fundamental limits of localizing change points in \Cref{assume:change point regression model}.  To the best of our knowledge, this is the first time such results are developed for \Cref{assume:change point regression model}.  The lower bounds on the localisation and detection nearly match the upper bounds we obtained under mild conditions.  
\item We present extensive experimental results including simulated data and real data analysis, supporting our theoretical findings, and confirming the practicality of our procedures.
\end{itemize}


Throughout this paper, we adopt the following notation.  For any set $S$, $|S|$ denotes its cardinality.  For any vector $v$, let $\|v\|_2$, $\|v\|_1$, $\|v\|_0$ and $\|v\|_{\infty}$ be its $\ell_2$-, $\ell_1$-, $\ell_0$- and entry-wise maximum norms, respectively; and let $v(j)$ be the $j$th coordinate of $v$.  For any square matrix $A \in \mathbb{R}^{n \times n}$, let $\Lambda_{\min}(A)$ and $ \Lambda_{\max}(A)$ be the smallest and largest eigenvalues of matrix $A$, respectively.    For any pair of integers $s, e \in \{0, 1, \ldots, n\}$ with $s < e$, we let $(s, e] = \{s + 1, \ldots, e\}$ and $[s, e] = \{s , \ldots, e\}$ be the corresponding integer intervals.

\section{METHODS}

\subsection{A Dynamic Programming Approach}

To achieve the goal of obtaining consistent change point estimators, we adopt a dynamic programming approach, whiuch we summarize next.  Let $\mathcal{P}$ be an integer interval partition of $\{1, \ldots, n\}$ into $K_{\cal P}$ intervals, i.e.
	\begin{align*}
		\mathcal{P} = \big\{\{1, \ldots, i_1-1\}, \{i_1, \ldots, i_2-1\}, \ldots, \{i_{K_{\mathcal{P}}-1}, \\
		\ldots, i_{K_{\mathcal{P}}} - 1\}\big\},
	\end{align*}
	for some integers $1 < i_1 < \cdots < i_{K_{\mathcal{P}}} = n + 1$, where $K_{\mathcal{P}} \geq 1$.  For a positive tuning parameter $\gamma > 0$, let 
	\begin{equation}\label{eq-wide-p}
		\widehat{\mathcal{P}} \in \argmin_{\mathcal{P}} \left\{\sum_{I \in \mathcal{P}} \mathcal{L}(I) + \gamma |\mathcal{P}|\right\}, 
	\end{equation}
	where $\mathcal{L}(\cdot)$ is an appropriate loss function to be specified below, $|\mathcal{P}|$ is the cardinality of $\mathcal{P}$ and the minimization is taken over all possible interval partitions of $\{1, \ldots, n\}$.
	
The change point estimator resulting from the solution to \eqref{eq-wide-p} is simply obtained by taking all the left endpoints of the intervals $I \in \widehat{\mathcal{P}}$, except $1$.  The optimization problem \eqref{eq-wide-p} is known as the \emph{minimal partition problem} and can be solved using dynamic programming with an overall computational cost of order $O(n^2 \mathcal T(n))$, where $\mathcal T(n)$ denotes the computational cost of solving $\mathcal{L}(I)$ with $|I| = n$ \citep[see e.g.~Algorithm 1 in][]{FriedrichEtal2008}.

To specialize the dynamic programming algorithm to the  high-dimensional linear regression model of interest  by setting the loss function to be
\eqref{eq-wide-p} with
	\begin{equation}\label{eq-loss-1}
		\mathcal L(I) = \sum_{t \in I}(y_t - x_t^{\top} \widehat \beta^\lambda_I)^2,
	\end{equation}
	where
	\begin{align}
		\widehat \beta^\lambda_I = \argmin_{v \in \mathbb R^p} \bigg\{\sum_{t \in I} (y_t - x_t^{\top}v)^2  \nonumber \\
		+ \lambda \sqrt{\max\{|I|, \, \log(n \vee p)\}} \| v\|_1\bigg\},\label{eq-beta-1}
	\end{align}
	and $\lambda \geq 0$ is a tuning parameter.  The penalty is multiplied by the quantity $\max\{|I|, \, \log(n \vee p)\}$ in order to fulfill certain types of large deviation inequalities that are needed to ensure consistency.  

Algorithms based on dynamic programming are widely used in the change point detection literature.  \cite{FriedrichEtal2008}, \cite{KillickEtal2012}, \cite{Rigaill2010}, \cite{MaidstoneEtal2017}, \cite{wang2018univariate}, among others, studied dynamic programming approaches for change point analysis involving a univariate time series with piecewise-constant means. \cite{leonardi2016computationally} examined high-dimensional linear regression change point detection problems by using a version of dynamic programming approach.

\subsection{Local Refinement}\label{sec-lr}

We will show later in \Cref{eq:DP consistent for regression} that the localization error afforded by the dynamic programming approach in \eqref{eq-wide-p}, \eqref{eq-loss-1}, and \eqref{eq-beta-1} is linear in $K$, the number of change points.  Although the corresponding localization rate is already sharper than any other rates previously established in the literature (see \Cref{sec-comparisons}), it is possible to improve it by removing the dependence on $K$ through an additional step, which we refer to as local refinement, detailed in \Cref{algorithm:LR}.

\begin{algorithm}[htbp]
\begin{algorithmic}
	\INPUT Data $\{(x_t, y_t)\}_{t=1}^{n}$, $\{\widetilde{\eta}_k\}_{k = 1}^{\widetilde{K}}$ ,  $\zeta > 0$.
	\State $(\widetilde{\eta}_0, \widetilde{\eta}_{\widetilde{K} + 1}) \leftarrow (0, n)$
	\For{$k = 1, \ldots, \widetilde{K}$}  
		\State $(s_k, e_k) \leftarrow (2\widetilde{\eta}_{k-1}/3 + \widetilde{\eta}_{k}/3, \widetilde{\eta}_{k}/3 + 2\widetilde{\eta}_{k+1}/3)$
		\State 
		\begin{align}
			& \left(\widehat{\beta}_1, \widehat{\beta}_2, \widehat{\eta}_k\right) \leftarrow \argmin_{\substack{\eta \in \{s_k + 1, \ldots, e_k - 1\} \\ \beta_1, \beta_2 \in \mathbb{R}^{p},\, \beta_1 \neq \beta_2}}   \Bigg\{\sum_{t = s_k + 1}^{\eta}\bigl\|y_{t} - \beta_1^{\top} x_t\bigr\|^2_2 \nonumber \\
			& \hspace{0.5cm}+ \sum_{t = \eta + 1}^{e_k}\bigl\|y_t - \beta_2 x_t\bigr\|_2^2 \nonumber \\
		    & \hspace{0.5cm} + \zeta  \sum_{i = 1}^p \sqrt{(\eta - s_k)(\beta_1)_{i}^2 + (e_k - \eta)(\beta_2)_{i}^2}\Bigg\} \label{eq-g-lasso}
		\end{align}
	\EndFor
	\OUTPUT $\{\widehat{\eta}_k\}_{k = 1}^{\widetilde{K}}$.
\caption{Local Refinement. }
\label{algorithm:LR}
\end{algorithmic}
\end{algorithm} 
	
\Cref{algorithm:LR} takes a sequence of preliminary change point estimators $\{\widetilde{\eta}_k\}_{k = 1}^{\widetilde{K}}$, and refines each of the estimator $\widetilde{\eta}_k$ within the interval $(s_k, e_k)$, which is a shrunken version of $(\widetilde{\eta}_{k-1}, \widetilde{\eta}_{k+1})$ (the constants $2/3$ and $1/3$ specifying the shrinking factor in the definition of $(\widetilde{\eta}_{k-1}, \widetilde{\eta}_{k+1})$ are not special and can be replaced by other values without affecting the rates).  The shrinkage is applied to eliminate false positives, which are more likely to occur in the immediate proximity of a preliminary estimate of a change point.  Since the refinement is done locally within each disjoint interval, the procedure is parallelizable.  A group Lasso penalty is deployed in~\eqref{eq-g-lasso}, and this is key to the success of the refinement.  Intuitively, the group Lasso penalty integrates the information that the regression coefficients are piecewise-constant within each coordinate.  Previously, \cite{wang2019statistically} also proposed a similar  local screening algorithm based on the group Lasso estimators to refine the estimates of the  regression change points.   While  \cite{wang2019statistically} assumed all the covariates to be uniformly bounded,  we show that \Cref{algorithm:LR} can achieve optimal localization error rates  in a more general setting.

\section{MAIN RESULTS}\label{sec-theory}
In this section, we derive  high-probability bounds on the localization errors of our main procedure based on the dynamic programming algorithm as detailed in equations \ref{eq-wide-p}, \ref{eq-loss-1}, and \ref{eq-beta-1}, and of the local refinement procedure of  \Cref{algorithm:LR}.

\subsection{Assumptions}
 
 We begin by stating the assumptions we require in order to derive localization error bounds.
 
\begin{assume} \label{assume:high dim coefficient}  
	Consider the model defined in \Cref{assume:change point regression model}. We assume that, for some fixed positive constants $C_{\beta}$, $c_x$, $C_x$, $\xi$, and   $C_{\mathrm{SNR}} $ the following holds:

\noindent \textbf{a.} (Sparsity).  Let $d_0 = |S|$.  There exists a subset $S \subset \{1, \ldots, p\}$ such that 
			\[
				\beta_t^*(j) = 0, \quad t = 1, \ldots, n, \quad j \in S^c = \{1, \ldots, p\} \setminus S.
			\]

\noindent \textbf{b.} (Boundedness). For some absolute constant $C_{\beta} > 0$, $\max_{t = 1, \ldots, n} \|\beta_t^*\|_{\infty} \leq C_{\beta}$.
	
\noindent \textbf{c.} (Minimal eigenvalue).  We have that $\Lambda_{\min}(\Sigma) = c_x^2 > 0$  and $\max_{j = 1, \ldots, p} (\Sigma)_{jj} = C_x^2 > 0$.

\noindent \textbf{d.} (Signal-to-noise ratio).  Let $\kappa= \min_{k = 1, \ldots, K + 1} \| \beta^*_{\eta_k} - \beta^*_{\eta_{k-1}}\|_2$ and $\Delta= \min_{k = 1, \ldots, K + 1} (\eta_{k} - \eta_{k-1})$ be the minimal jump size and minimal spacing defined as follows, respectively. Then,
			\begin{equation}\label{eq:snr}
				\Delta \kappa^2 \geq C_{\mathrm{SNR}} d_0^2 K \sigma^2_{\varepsilon}\log^{1 + \xi}(n\vee p).
			\end{equation}
\end{assume}

\Cref{assume:high dim coefficient}(\textbf{a}) and (\textbf{c}) are standard conditions required for consistency of Lasso-based estimators.  \Cref{assume:high dim coefficient}(\textbf{d}) specifies a minimal  signal-to-noise ratio condition that allows to detect the presence of a change point.  Interestingly, if $K = d_0 = 1$, \eqref{eq:snr} reduces to $\Delta \kappa^2 \sigma^{-2}_{\epsilon}\gtrsim \log^{1 + \xi}(n \vee p)$, matching the information theoretic lower bound  (up to constants and logarithmic terms) for the univariate mean change point detection problem \citep[see e.g.][]{chan2013, FrickEtal2014, wang2018univariate}.
			
In addition, we have
	\begin{align}\label{eq-eff-size}
		\Delta & \geq \frac{C_{\mathrm{SNR}} d_0^2 K \sigma^2_{\epsilon} \log^{1 + \xi}(n \vee p)}{\kappa^2} \nonumber \\
		& \geq 	\frac{C_{\mathrm{SNR}} d_0^2 K \sigma^2_{\epsilon} \log^{1 + \xi}(n \vee p)}{4C_{\beta}^2 d_0} \nonumber \\
		& \geq \frac{C_{\mathrm{SNR}}}{4C_{\beta}^2} d_0 K \sigma^2_{\epsilon} \log^{1+\xi}(n \vee p),
	\end{align}
	where the second inequality follows from the bound
	\[
		\kappa^2 = \min_{k = 1, \ldots, K + 1} \| \beta^*_{\eta_k} - \beta^*_{\eta_{k-1}}\|_2^2 \leq d_0 (2C_{\beta})^2 = 4C_{\beta}^2 d_0.
	\]
   If $\Delta = \Theta (n)$ and $K=O(1)$, then \eqref{eq-eff-size}  becomes $ n \gtrsim d_0 \log^{1+\xi}(n \vee p)$, 
   which  resembles the effective sample size condition needed in the Lasso estimation literature.
	
Another way to interpret the signal-to-noise ratio \Cref{assume:high dim coefficient}(\textbf{d}) is to introduce a normalized jump size $\kappa_0 = \kappa/\sqrt{d_0}$, which leads to the equivalent condition
	\[
		\Delta \kappa_0^2 \geq C_{\mathrm{SNR}} d_0 K \sigma^2_{\epsilon} \log^{1+\xi}(n \vee p).
	\]
	Analogous constraints on the model parameters are required in other change point detection problems, including high-dimensional mean change point detection \citep{wang2018high}, high-dimensional covariance change point detection \citep{wang2017optimal}, sparse dynamic network change point detection \citep{wang2018optimal}, high-dimensional regression change point detection \citep{wang2019statistically}, to name but a few.  Note that in these aforementioned papers, when variants of wild binary segmentation \citep{fryzlewicz2014wild} were deployed, additional knowledge is needed to get rid of $K$ in the lower bound of the signal-to-noise ratio.  We refer the reader to \cite{wang2018optimal} for more discussions regarding this point.
	
The constant $\xi$ is needed to guarantee consistency when $\Delta$ is of the same irder as $n$ but  can be set to zero if $\Delta = o(n)$.  We may instead replace it with a weaker condition of the form
	\[
		\Delta \kappa^2 \gtrsim C_{\mathrm{SNR}} d_0^2 K \{\log(n \vee p) + a_n\},
	\]
	where $a_n \to \infty$ arbitrarily slow as $n \to \infty$.  We stick with the signal-to-noise ratio condition \eqref{eq:snr} for simplicity.

\subsection{Localization Rates}\label{sec-local-rates}
We are now ready to state one of the main results of the paper.

\begin{theorem}\label{eq:DP consistent for regression}
Assume \Cref{assume:change point regression model} and the conditions in  \Cref{assume:high dim coefficient}. Then,
the change point estimators $\{\widetilde{\eta}_k\}_{k = 1}^{\widehat{K}}$ obtained as a solution to the dynamic programming optimization problem given in \eqref{eq-wide-p}, \eqref{eq-loss-1}, and \eqref{eq-beta-1}  with tuning parameters
	\[
		\lambda = C_{\lambda}\sigma_{\varepsilon} \sqrt{d_0 \log(n \vee p)}
	\]
	and
	\[
		\gamma = C_{\gamma}\sigma_{\epsilon}^2 (K+1) d^2_0 \log(n \vee p),
	\]
	are such that
	\begin{align}\label{eq-lm-thm-1}
		& \mathbb{P}\left\{\widehat K = K, \, \max_{k = 1, \ldots, K}|\widetilde{\eta}_k - \eta_k| \leq  \frac{KC_{\epsilon}d_0^2 \sigma^2_{\epsilon} \log(n \vee p)}{\kappa^2} \right\} \nonumber\\ 
		& \hspace{1cm}\geq 1 - C(n \vee p)^{-c},
	\end{align}
	where $C_{\lambda}, C_{\gamma}, C_{\epsilon}, C, c > 0$ are absolute constants depending only on $C_{\beta}, C_x$, and $c_x$.
\end{theorem}

The above result implies that, with probability tending to 1 as $n$ grows,
	\begin{align*}
		\max_{k = 1, \ldots, K}\frac{|\hat{\eta}_k - \eta_k|}{\Delta} \leq \frac{KC_{\epsilon} d^2_0 \sigma^2_{\epsilon}\log(n \vee p)}{\kappa^2 \Delta} \\
		\leq \frac{C_{\epsilon}}{C_{\mathrm{SNR}}\log^{\xi}(n \vee p)} \to 0,
	\end{align*}
	where in the second inequality we have used \Cref{assume:high dim coefficient}(\textbf{c}).  Thus, the localization error converges to zero in probability.
	
	It is worth emphasizing that the bound in \eqref{eq-lm-thm-1} along with \Cref{assume:change point regression model}   provide a {\it family} of rates, depending on how the model parameters ($p$, $d_0$, $\kappa$, $\Delta$, $K$ ands $\sigma_{\epsilon}$) scale with $n$.
	
The tuning parameter $\lambda$ affects the performance of the Lasso estimator.  The second tuning parameter $\gamma$ prevents overfitting while searching the optimal partition as a solution to the problem \eqref{eq-wide-p}.  In particular, $\gamma$ is determined by the squared $\ell_2$-loss of the Lasso estimator and is of order $\lambda^2 d_0$.  We will elaborate more on this point in the supplementary materials.	
	
We now turn to the analysis of the local refinement algorithm, which takes as input a preliminary collection of change point estimators $\{\widetilde{\eta}_k\}_{k = 1}^{K}$ such that $\max_{k = 1, \ldots, K} |\widetilde{\eta}_k - \eta_k| $, such as the ones returned by our estimator based on the dynamic programming approach. The only assumption required for local refinement  is that the localization error of  the preliminary estimators be a small enough fraction of the minimal spacing $\Delta$ (see \eqref{eq-lr-cond-1} below). Then local refinement returns an improved collection of change point estimators $\{\widehat{\eta}_k\}_{k = 1}^{K}$ with a vanishing localization error rate of order $O \left( \frac{d_0 \log ( n \vee p )}{n \kappa^2} \right)$. Interestingly, the initial estimators need not be consistent in order for local refinement to work: all that is required is essentially that the each of the working intervals in \Cref{algorithm:LR} contains one and only one true change point. This fact allows us to refine the search within each working intervals separately, yielding better rates. 

In particular, if we use the outputs of \eqref{eq-wide-p}, \eqref{eq-loss-1}, and \eqref{eq-beta-1} as the inputs of \Cref{algorithm:LR}, then it follows from \eqref{eq-lm-thm-1} and \eqref{eq:snr} that, for any $k \in \{1, \ldots, K\}$, 
	\begin{align*}
		s_k - \eta_{k - 1} > \frac{2}{3}\widetilde{\eta}_{k - 1} + \frac{1}{3}\widetilde{\eta}_{k} - \widetilde{\eta}_{k - 1} - \epsilon \\
		= \frac{1}{3}(\widetilde{\eta}_k - \widetilde{\eta}_{k-1}) - \epsilon > \Delta/3 - 5\epsilon/3 > 0
	\end{align*}
	and
	\begin{align*}
		s_k - \eta_{k} < \frac{2}{3}\widetilde{\eta}_{k - 1} + \frac{1}{3}\widetilde{\eta}_{k} - \widetilde{\eta}_{k} + \epsilon \\
		= - \frac{2}{3}(\widetilde{\eta}_{k-1} - \widetilde{\eta}_{k}) + \epsilon < -2\Delta/3 + 5\epsilon/3 < 0.
	\end{align*}

\begin{corollary}\label{cor-lr-3}
Assume the same conditions of \Cref{eq:DP consistent for regression}.
Let $\{\widetilde{\eta}_k\}_{k = 1}^{K}$ be a set of time points satisfying
	\begin{equation}\label{eq-lr-cond-1}
		\max_{k = 1, \ldots, K} |\widetilde{\eta}_k - \eta_k| \leq \Delta/7.
	\end{equation}
	Let $\{\widehat{\eta}_k\}_{k = 1}^{\widehat{K}}$ be the change point estimators generated from \Cref{algorithm:LR} with $\{\widetilde{\eta}_k\}_{k = 1}^{K}$ and 
	\[
		\zeta = C_{\zeta}\sqrt{\log(n \vee p)}
	\]
	as inputs.  Then,
	\begin{align*}
		& \mathbb{P}\left\{\widehat K = K, \, \max_{k = 1, \ldots, K}|\hat{\eta}_k - \eta_k| \leq \frac{C_{\epsilon} d_0 \log(n \vee p)}{\kappa^2} \right\} \\
		& \hspace{1cm}\geq 1 - Cn^{-c},
	\end{align*}
	where $C_{\zeta}, C_{\epsilon}, C, c > 0$ are absolute constants depending only on $C_{\beta}, \mathcal{M}$ and $c_x$.
\end{corollary}

Compared to the localization error given in \Cref{eq:DP consistent for regression}, the improved localization error guaranteed by the local refinement algorithm does not have a direct dependence on $K$, the number of change points. The intuition for this is as follows.
First, due to the nature of the change point detection problem, there is a natural group structure.  This justifies the use of the group Lasso-type penalty, which reduces the localization error by bringing down $d_0^2$ to $d_0$. Second, using condition \eqref{eq-lr-cond-1}, there is one and only one true change point in every working interval used by the local refinement algorithm. The true change points can then be estimated separately using $K$ independent searches, in such a way that the final localization rate that does not depend on the number of searches, namely $K$.

\subsection{Comparisons}\label{sec-comparisons}

We now discuss how our contributions compared with the results of \cite{wang2019statistically} and of \cite{leonardi2016computationally}, which investigate the same high-dimensional change-point linear regression model.

\cite{wang2019statistically} proposed different algorithms, all of which are variants of wild binary segmentation, with or without additional Lasso estimation procedures.  Those methods inherit both the advantages and the disadvantages of WBS.  Compared with dynamic programming, WBS-based methods require additional tuning parameters such as randomly selected intervals as inputs.  With these additional tuning parameters, Theorem~1 in \cite{wang2019statistically} achieved the same statistical accuracy in terms of the localization error rate as   \Cref{eq:DP consistent for regression} above. 
		In terms of computational cost, the methods in \cite{wang2019statistically} are of order $O(K^2 n\cdot \text{Lasso}(n))$, where $K$, $n$ and $\text{Lasso}(n)$ denote the number of change points, the sample size and the computational cost of Lasso algorithm with sample size $n$, respectively, while the dynamic programming approach of this paper   is of order $O(n^2 \cdot \text{Lasso}(n))$.  Thus, when $K \lesssim \sqrt{n}$, the algorithm in \cite{wang2019statistically} is computationally more efficient, but when $K \gtrsim \sqrt{n}$, the method in this paper has smaller complexity.

\cite{leonardi2016computationally} analysed two algorithms, one based on a dynamic programming approach, and the other on binary segmentation, and claimed that they both yield the same localization, which is, in our notation,
		\begin{equation}\label{eq-bul}
			\sum_{k = 1}^K |\hat{\eta}_k - \eta_k| \lesssim \frac{d_0^2 \sqrt{n \log(np)}}{\kappa^2}.
		\end{equation}

Note that, the error bound  in \cite{leonardi2016computationally}   is originally of the form  $	\sum_{k = 1}^K |\hat{\eta}_k - \eta_k| \lesssim \frac{d_0 \sqrt{n \log(np)}}{\kappa^2}$ under a slightly stronger assumption than ours.  In the more general settings of \Cref{assume:high dim coefficient},  the localization error bound  of    \cite{leonardi2016computationally} is of the form \eqref{eq-bul}, based on personal communication with the authors.  
		
	 	It is not immediate to directly compare the sum of all localization errors, used by \cite{leonardi2016computationally}, with the maximum localization error, which is the target in this paper.  Using a worst-case upper bound, \Cref{eq:DP consistent for regression} yields that 
	 	\[
	 		\sum_{k = 1}^K |\hat{\eta}_k - \eta_k| \lesssim \frac{K^2 d_0^2 \sigma^2_{\varepsilon} \log(n\vee p)}{\kappa^2}.
	 	\]
	 	In light of \Cref{cor-lr-3}, this error bound can be sharpened, using the local refinement \Cref{algorithm:LR} to
	 	\[
	 		\sum_{k = 1}^K |\hat{\eta}_k - \eta_k| \lesssim \frac{K d_0 \sigma^2_{\varepsilon} \log(n\vee p)}{\kappa^2}.
	 	\]
	 	As long as $K^2 \lesssim \sqrt{\frac{n}{\log(np)}}$, or, using the local refinement algorithm,  $K\lesssim \sqrt{\frac{n}{\log(np)}}$, our localization rates are better than the one implied by \eqref{eq-bul}.
	 	It is not immediate to compare directly the assumptions used in \cite{leonardi2016computationally} with the ones we formulate here due to the different ways we use to present them.  For instance, the conditions in Theorem~3.1 of \cite{leonardi2016computationally} imply, in our notation, that condition 
	 		\[
	 			\Delta \gtrsim \sqrt{n \log(p)}
	 		\]
	 		is needed for consistency,
	 		even if the sparsity parameter $d_0 = \Theta(1)$.  However in our case, in view of \eqref{eq:snr}, if we assume $d_0 = \kappa = \Theta(1)$, then we only require $\Delta \gtrsim \log^{1 + \xi}(n\vee p)$ for consistency.

\subsection{Lower Bounds}

In \Cref{sec-local-rates}, we show that as long as
	\[
		\kappa^2 \Delta \gtrsim d_0^2 K \sigma^2_{\varepsilon} \log^{1+\xi}(n \vee p), 
	\]
	we demonstrate provide change point estimators with localization errors upper bounded by
	\[
		d_0 \sigma^2_{\varepsilon} \kappa^{-2} \log(n \vee p).
	\]
	In this section, we show that no algorithm is guaranteed to be consistent in the regime
	\[
		\kappa^2 \Delta \lesssim d_0 \sigma^2_{\varepsilon},
	\]
	and otherwise, a minimax lower bound on the localization errors is
	\[
		d_0 \sigma^2_{\varepsilon} \kappa^{-2}.
	\]
These findings are formally stated next, in  Lemmas~\ref{lem-reg-lb-1} and \ref{lem-reg-lb-2}, respectively.

\begin{lemma}\label{lem-reg-lb-1}
Let $\{(x_t, y_t)\}_{t = 1}^T \subset \mathbb{R}^p \times \mathbb{R}$ satisfy \Cref{assume:change point regression model} and \Cref{assume:high dim coefficient}, with $K = 1$.  In addition, assume that $\{x_t\}_{t=1}^n \stackrel{\mbox{iid}}{\sim} \mathcal{N}(0, I_p)$ and $\{ \varepsilon_t \}_{t=1}^n \stackrel{\mbox{iid}}{\sim} \mathcal{N}(0, \sigma_{\varepsilon}^2)$.  Let $P^T_{\kappa, \Delta, \sigma_{\varepsilon}, d}$ be the corresponding joint distribution.  For any $0 < c < \frac{2}{8e+1}$, consider the class of distributions
	\begin{align*}
		\mathcal{P} = \Big\{P^T_{\kappa, \Delta, \sigma_{\varepsilon}, d}: \, \Delta = \min\left\{\lfloor cd_0 \sigma_{\varepsilon}^2\kappa^{-2}\rfloor, \, \lfloor T/4\rfloor \right\},\\
		 2cd_0\max\{d_0, 2\} \leq \Delta\Big\}.
	\end{align*}
	There exists a $T(c)$, which depends on $c$, such that for all $T \geq T(c)$,
	\[
		\inf_{\widehat{\eta}} \sup_{P \in \mathcal{P}} \mathbb{E}_P(|\widehat{\eta} - \eta(P)|) \geq \Delta,
	\]
	where $\eta(P)$ is the location of the change point of distribution $P$ and the infimum is over all estimators of the change point. 
\end{lemma}

\Cref{lem-reg-lb-1} shows that if $\kappa^2 \Delta \lesssim d_0 \sigma^2_{\varepsilon}$, then 
	\[
		\frac{\inf_{\widehat{\eta}} \sup_{P \in \mathcal{P}} \mathbb{E}_P(|\widehat{\eta} - \eta(P)|)}{\Delta} \geq 1,
	\]
	which implies that the localization error is not a vanishing fraction of $\Delta$ as the sample size grows unbounded.

\begin{lemma}\label{lem-reg-lb-2}
Let $\{(x_t, y_t)\}_{t = 1}^T \subset \mathbb{R}^p \times \mathbb{R}$ satisfy \Cref{assume:change point regression model} and \Cref{assume:high dim coefficient}, with $K = 1$.  In addition, assume $\{x_t\}_{t=1}^n \stackrel{\mbox{iid}}{\sim} \mathcal{N}(0, I_p)$ and $\{ \varepsilon_t \}_{t=1}^n \stackrel{\mbox{iid}}{\sim} \mathcal{N}(0, \sigma^2_{\varepsilon})$.  Let $P^T_{\kappa, \Delta, \sigma_{\varepsilon}, d}$ be the corresponding joint distribution.  For any diverging sequence $\zeta_T$, consider the class of distributions
	\[
		\mathcal{P} = \left\{P^T_{\kappa, \Delta, \sigma_{\varepsilon}, d}: \, \Delta = \min\left\{\lfloor \zeta_T d_0 \sigma_{\varepsilon}^2\kappa^{-2}\rfloor, \, \lfloor T/4\rfloor \right\} \right\}.
	\]
	Then 
	\[
		\inf_{\widehat{\eta}} \sup_{P \in \mathcal{P}} \mathbb{E}_P(|\widehat{\eta} - \eta(P)|)\geq \frac{cd_0 \sigma_{\varepsilon}^2}{\kappa^2},
	\]
	where $\eta(P)$ is the location of the change point of distribution $P$, the infimum is over all estimators of the change point and $c > 0$ is an absolute constant. 
\end{lemma}

Recalling all the results we have obtained, the change point localization task is either impossible when $\kappa^2 \Delta \lesssim d_0 \sigma^2_{\varepsilon}$ (in the sense that  no algorithm is guaranteed to be consistent) or, when $\kappa^2 \Delta \gtrsim d_0^2 K \sigma^2_{\varepsilon} \log^{1+\xi}(n \vee p)$, it can be solved by our algorithms at nearly a minimax optimal rate. 

In the intermediate case 
\[
d_0 \sigma^2_{\varepsilon} \lesssim \kappa^2 \Delta \lesssim d^2_0 K \sigma^2_{\varepsilon} \log^{1+\xi}(n \vee p) 
\]
we are unable to provide a result one way or another. 	
However, we remark that, if in addition, in \Cref{eq:DP consistent for regression}, we assume $\kappa \leq C$, for an absolute constant $C > 0$, then we are able to weaken the condition from $\kappa^2 \Delta \gtrsim d_0^2 K \sigma^2_{\varepsilon} \log^{1+\xi}(n \vee p)$ to $\kappa^2 \Delta \gtrsim d_0 K \sigma^2_{\varepsilon} \log^{1+\xi}(n \vee p)$, by almost identical arguments.  This shows that, under the additional conditions 
	\[
		\max\{\kappa, \, K\} \leq C,
	\]
	for an absolute constant $C > 0$, the condition is nearly optimal, off by a logarithmic factor.


\section{NUMERICAL EXPERIMENTS}

In this section, we investigate the numerical performances of our proposed methods, with efficient binary segmentation algorithm (EBSA) of \cite{leonardi2016computationally} as the competitor.  We compare four methods: dynamic programming (DP, see \ref{eq-wide-p}, \ref{eq-loss-1}, and \ref{eq-beta-1}), EBSA, local refinement (\Cref{algorithm:LR}) initialized by DP (DP.LR), and local refinement (\Cref{algorithm:LR}) initialized by EBSA (EBSA.LR).  

The evaluation metric considered is the scaled Hausdorff distance between the estimators $\{\hat{\eta}_k\}_{k = 1}^{\hat K}$ and the truth $\{\eta_k\}_{k =1}^K$.  To be specific, we report
$\mathit{d}(\hat{\mathcal{C}},\mathcal{C}) = n^{-1}\mathcal{D}(\hat{\mathcal{C}},\mathcal{C})$, where 
    \[
        \mathcal{D}(\hat{\mathcal{C}},\mathcal{C}) = \max\{\max_{\hat{\eta} \in \hat{\mathcal{C}}}\min_{\eta \in \mathcal{C}}|\hat{\eta} - \eta|,\max_{\eta \in \mathcal{C}}\min_{\hat{\eta} \in \hat{\mathcal{C}}}|\hat{\eta} - \eta|\},
    \]
     $\mathcal{C} = \{\eta_k\}_{k =1}^K$ and $\hat{\mathcal{C}} = \{\hat{\eta}_k\}_{k = 1}^{\hat K}$.

We consider both simulated data and a real-life public dataset on air quality indicators in Taiwan. 

\subsection{Tuning Parameter Selection}
We adopt a cross-validation approach to choosing tuning parameters.  Let samples with odd indices be the training set and even ones be the validation set.  Recall that for the DP, we have two tuning parameters $\lambda$ and $\gamma$, which we tune using a brute-force grid search.  For each pair of tuning parameters, we conduct DP on the training set and obtain estimated change points.  Within each estimated segment of the training set, we obtain $\hat{\beta_t}$ by \eqref{eq-beta-1}.  On the validation set, let $\hat{y_t} = x_t^{\top} \hat{\beta}_t$ and calculate the validation loss $(n/2)^{-1} \sum_{t \mod 2 \equiv 0}(\hat{y}_t - y_t)^2$.  The pair $(\lambda, \gamma)$ is chosen to be the one corresponding to the lowest validation loss. 

As for the simulated data, we use some prior knowledge of the truth to save some computational cost.  To be specific, we let the odd index set be partitioned by the true change points and estimate $\beta_t$ on these intervals.  We then plot the mean squared errors of $\hat{\beta}_t$ across a range of values of $\lambda$ and obtain an ``optimal'' $\lambda$.  We choose the grid range of $\lambda$ around the ``optimal'' $\lambda$.  This step is to approximately locate the range of $\lambda$'s value but this step will not be used in the real data experiment.  The same procedure is conducted for the tuning parameter selection in EBSA.

For the local refinement algorithm, we let the estimated change points of DP or EBSA be the initializers of the local refinement algorithm.  We then regard the initialization algorithm and local refinement as a self-contained method and tune all three parameters $\lambda$, $\gamma$, and $\zeta$ jointly.  The tuning procedure is almost the same as we described above, except that we use 
    \[
		\widehat \beta^\lambda_I = \argmin_{v \in \mathbb R^p} \bigg\{\sum_{t \in I} (y_t - x_t^{\top}v)^2  	+ \zeta \sqrt{|I|} \| v\|_1\bigg\} 
	\]
    to estimate $\beta_t$. 

\subsection{Simulations}

Throughout this section, we let $n = 600$, $p = 200$, $K = 4$, $\Sigma = I$ and $\sigma_{\epsilon} = 1$. The true change points are at 121, 221, 351 and 451.  Let $\beta_0 = (\beta_{0i}, i = 1, \ldots, p)^{\top}$, with $\beta_{0i} = 2^{-1}d_0^{-1/2}\kappa$, $i \in \{1, \ldots, d_0\}$, and zero otherwise.  Let
    \[
        \beta_t = \begin{cases} 
            \beta_0, & t \in \{1, \ldots, 120\}, \\
            -\beta_0, & t \in \{121, \ldots, 220\}, \\
            \beta_0, & t \in \{221, \ldots, 350\}, \\
            -\beta_0, & t \in \{351, \ldots, 450\},  \\
            \beta_0, & t \in \{451, \ldots, 600\}.
        \end{cases}
    \] 
    
We let $\kappa \in \{4, 5, 6\}$ and $d_0 \in \{10, 15, 20\}$.  For each pair of $\kappa$ and $d_0$, the experiment is repeated 100 times.  The results are reported in Table~\ref{error-table} and Figure \ref{error-plot}.

Generally speaking, DP outperforms EBSA, and LR significantly improves upon EBSA and DP when DP doesn't give accurate results.  LR is comparable with DP when the initial points estimated by DP are already good enough.  Note that since for EBSA.LR we tune EBSA to optimize EBSA.LR's performance, it may well happen that the estimated number of change points $K$ from EBSA.LR is much different than from EBSA.
\begin{table*}[h]
\caption{Scaled Hausdorff Distance. The numbers in the brackets indicate the corresponding standard errors of the scaled Hausdorff distance.}
\label{error-table}
    \begin{center}
    \begin{tabular}{c c c c c c}
         \hline\hline
         Setting & Cases & DP & DP.LR & EBSA & EBSA.LR \\
         \hline
         $\kappa = 4, d_0 = 10$ & & 0.023(0.015) &0.008(0.004) &0.104(0.031) &0.034(0.045)  \\
         $\kappa = 4, d_0 = 15$ & All &0.031(0.020) &0.017(0.047) &0.104(0.029) &0.038(0.050) \\
         $\kappa = 4, d_0 = 20$ & &0.038(0.032) &0.019(0.042) &0.104(0.027) & 0.036(0.051)\\
         \hline
         $\kappa = 4, d_0 = 10$ & & 0.022(0.015) &0.008(0.004) & 0.061(0.047) & 0.008(0.008)\\
         $\kappa = 4, d_0 = 15$ & $\hat{K} = K$ & 0.025(0.018) &0.008(0.007) &0.071(0.045) & 0.010(0.016)\\
         $\kappa = 4, d_0 = 20$ & &0.028(0.020) &0.014(0.012) &0.076(0.048) & 0.010(0.011)\\
         \hline
         $\kappa = 5, d_0 = 10$ & &0.022(0.022) & 0.007(0.004) &0.102(0.033) & 0.033(0.046)\\
         $\kappa = 5, d_0 = 15$ & All &0.025(0.023) &0.015(0.025) &0.102(0.030) & 0.027(0.042)\\
         $\kappa = 5, d_0 = 20$ & &0.030(0.027) &0.016(0.030) &0.102(0.030) & 0.041(0.048)\\
         \hline
         $\kappa = 5, d_0 = 10$ & &0.020(0.015) &0.007(0.004) &0.068(0.073) & 0.007(0.008)\\
         $\kappa = 5, d_0 = 15$ & $\hat{K} = K$ &0.021(0.012) &0.010(0.006) &0.075(0.049) & 0.007(0.008)\\
         $\kappa = 5, d_0 = 20$ & &0.025(0.018) &0.010(0.007) &0.076(0.065) & 0.010(0.012)\\
         \hline
         $\kappa = 6, d_0 = 10$ & &0.009(0.010) &0.007(0.004) &0.100(0.028) & 0.034(0.049)\\
         $\kappa = 6, d_0 = 15$ & All &0.022(0.017) &0.009(0.005) &0.101(0.029) & 0.037(0.049)\\
         $\kappa = 6, d_0 = 20$ & &0.023(0.017) &0.010(0.006) &0.102(0.031) & 0.028(0.043)\\
         \hline
         $\kappa = 6, d_0 = 10$ & &0.009(0.010) &0.007(0.004) &0.061(0.064) & 0.006(0.010)\\
         $\kappa = 6, d_0 = 15$ & $\hat{K} = K$ &0.022(0.017) &0.009(0.005) &0.064(0.050) & 0.007(0.007)\\
         $\kappa = 6, d_0 = 20$ & &0.023(0.017) &0.010(0.006) &0.076(0.041) & 0.009(0.013)\\
         \hline
    \end{tabular}
    \end{center}
\end{table*}

\begin{figure}[h]\label{error-plot}
 
\includegraphics[width=9cm]{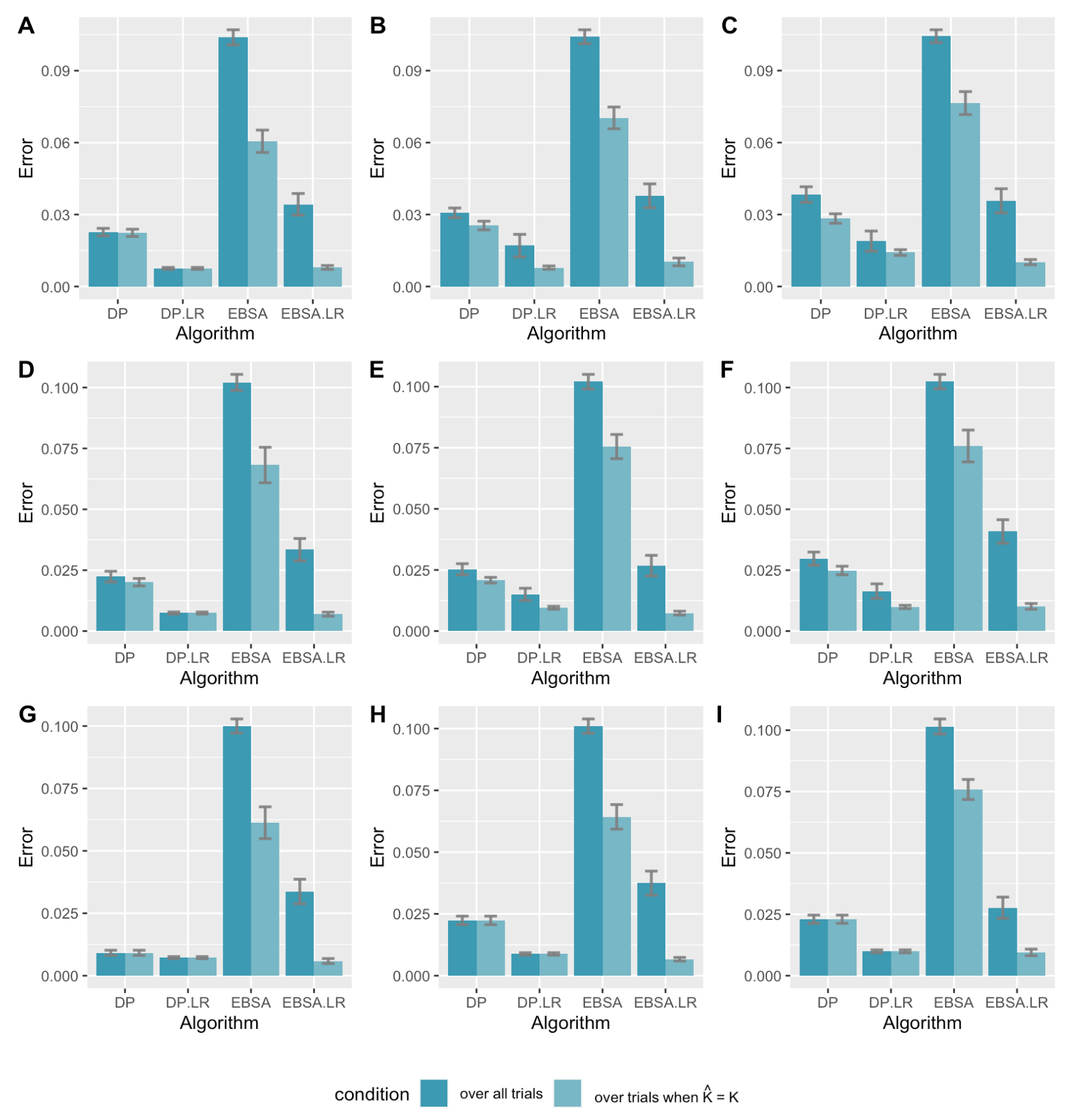}
\caption{Bar plots for results in Table \ref{error-table}. Plots A-C are settings with $\kappa = 4$ and $d_0 \in \{10,15,20\}$; Plots D-F are settings with $\kappa = 5$ and  $d_0 \in \{10,15,20\}$; and Plots G-I are settings with $\kappa = 6$ and  $d_0 \in \{10,15,20\}$.}
\end{figure}

\subsection{Air Quality Data}
In this subsection, we consider the air quality data from \url{https://www.kaggle.com/nelsonchu/air-quality-in-northern-taiwan}. It collects environment information and air quality data from Northern Taiwan in 2015.  We choose the PM10 in Banqiao as the response variables, with covariates being the temperature, the CO level, the NO level, the $\text{NO}_2$ level, the  $\text{NO}_x$ level, the rainfall quantity, the humidity quantity, the $\text{SO}_2$ level, the ultraviolet index, the wind speed, the wind direction and the PM10 levels in Guanyin, Longtan, Taoyuan, Xindian, Tamsui, Wanli and Keelung District.  
%
%
We transfer the hourly data into daily by averaging across 24 hours. After removing all dates containing missing values, we obtain a data set with $n = 343$ days and $p = 18$ covariates.  Our goal is to detect potential change points of this data set and determine if they are consistent with the historical information.

We standardize the data so that the variance of $y_t = 1$, for all $t$.  We then conduct DP, DP.LR, EBSA and EBSA.LR. DP estimates 2 change points which are March 16th and November 1st, 2015.  No change points are detected by EBSA. DP.LR and EBSA.LR both detect May 15th and October 25th, 2015 as the change points. 

The first change point detected by DP.LR and EBSA.LR seems to correspond with the first strong-enough typhoon near Northern Taiwan in 2015, which happened during May 6th-20th \citep[e.g.][]{wiki}.  The second change points from EBSA.LR, DP.LR and DP are relatively close and they all could be explained by the severe air pollution at the beginning of November in Taiwan, which reached the hazardous purple alert on November 8th \citep[e.g.][]{wikipollution}.  The visualization is shown in \Cref{fig-realdata}.

\begin{figure}[htbp]
\centerline{}
\includegraphics[width=0.45\textwidth]{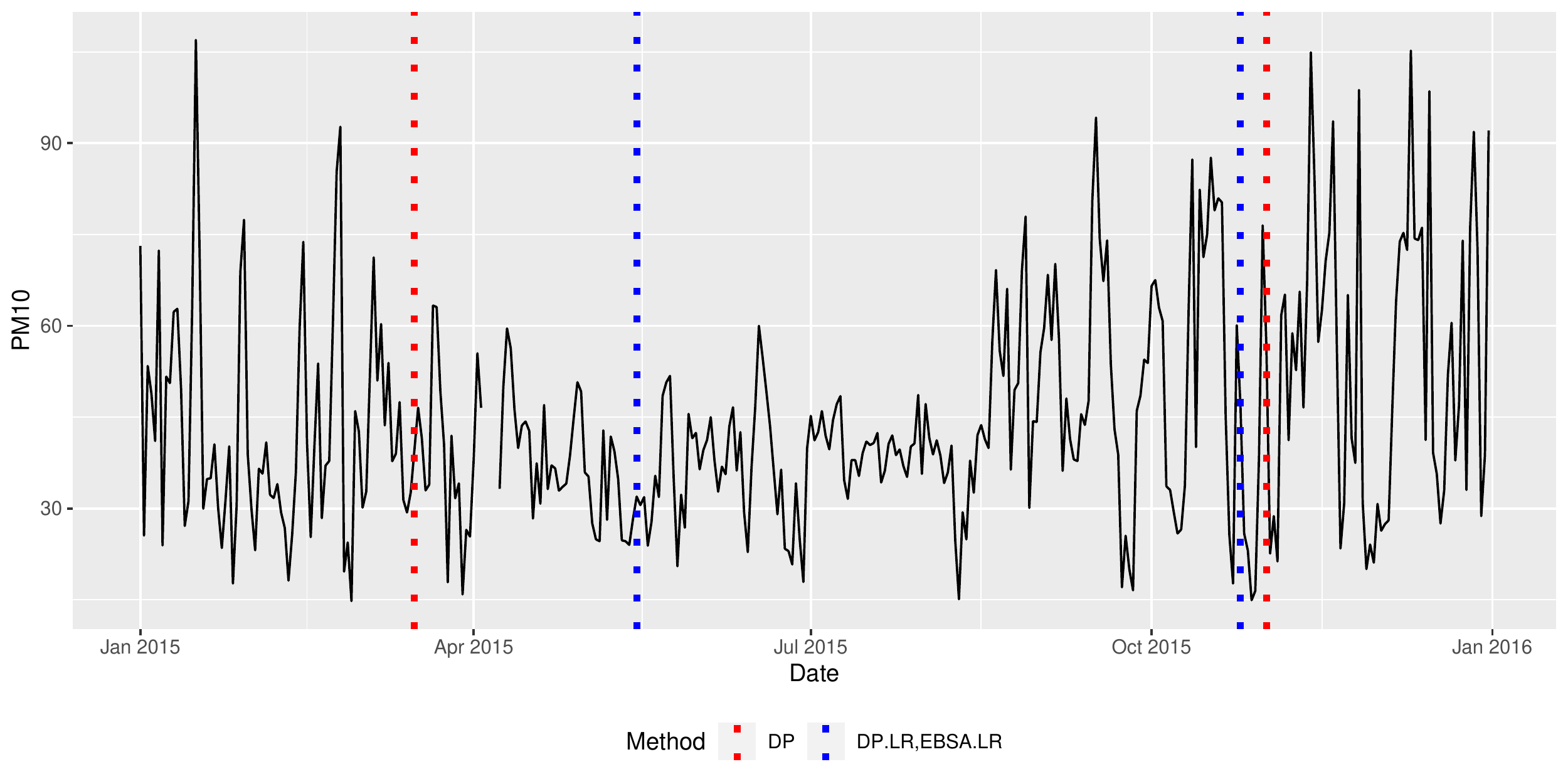}
\caption{Real air quality example.  When we tune EBSA for its standalone performance, we obtain zero change points, but when we tune EBSA to optimize the EBSA.LR's performance, we obtain two change points.\label{fig-realdata}}
\end{figure}

\section{DISCUSSION}

In this paper, we in fact provide a general framework for analyzing general regression-type change point localization problems that include the linear regression model above as a special case.   The analysis in this paper may be utilized as a blueprint for more complex change point localization problems.  In our analysis, we develop a new and refined toolbox for the change point detection community to study more complex data generating mechanisms above and beyond linear regression models.

\newpage

\bibliographystyle{authordate1}
\bibliography{citations}

\onecolumn
\aistatstitle{Localizing Changes in High-Dimensional Regression Models: \\
Supplementary Materials}

\section{Proof of Theorem~\ref{eq:DP consistent for regression}}\label{app-pf1}

\subsection{Sketch of the Proofs}\label{sec-rm}

In this subsection, we first sketch the proof of Theorem~\ref{eq:DP consistent for regression}, which serves as a general template to derive upper bounds on the localization error change point problems in the general regression framework described in \Cref{assume:change point regression model}. 

\Cref{eq:DP consistent for regression} is an immediate consequence of Propositions~\ref{proposition:dp step 1} and \ref{prop-2}.

\begin{proposition}\label{proposition:dp step 1}
Under the same conditions in \Cref{eq:DP consistent for regression} and letting $\widehat{\mathcal{P}}$ being the solution to \eqref{eq-wide-p}, the following hold with probability at least $1 - C(n \vee p)^{-c}$. 
	\begin{itemize}
		\item [(i)] For each interval  $\widehat I = (s, e] \in \widehat{\mathcal{P}}$    containing one and only one true change point $\eta$, it must be the case that
			\[
				\min\{e - \eta, \eta - s\} \leq C_{\epsilon} \left(\frac{d_0 \lambda^2 + \gamma}{\kappa^2}\right),
			\]
			where $C_{\epsilon} > 0$ is an absolute constant;
		\item [(ii)] for each interval  $\widehat I = (s, e] \in \widehat{\mathcal{P}}$ containing exactly two true change points, say  $\eta_1 < \eta_2$, it must be the case that
			\[
				\max\{e - \eta_2, \eta_1 - s\} \leq C_{\epsilon} \left(\frac{d_0 \lambda^2 + \gamma}{\kappa^2}\right), 	
			\]  
			where $C_{\epsilon} > 0$ is an absolute constant;
		\item [(iii)] for all consecutive intervals $\widehat I$ and $\widehat {J}$ in $\widehat{P}$, the interval $\widehat I \cup \widehat{J}$ contains at least one true change point; and
		\item [(iv)] no interval $\widehat I \in \widehat{\mathcal{P}}$ contains strictly more than two true change points.
	 \end{itemize}
\end{proposition} 


\begin{proposition}\label{prop-2}
Under the same conditions in \Cref{eq:DP consistent for regression}, with $\widehat{\mathcal{P}}$ being the solution to \eqref{eq-wide-p}, satisfying $K \leq |\widehat{\mathcal{P}}| \leq 3K$, then with probability at least $1 - C(n \vee p)^{-c}$, it holds that $|\widehat{\mathcal{P}}| = K$.	
\end{proposition}

\begin{proof}[Proof of \Cref{eq:DP consistent for regression}]
	It follows from \Cref{proposition:dp step 1} that, $K \leq |\widehat{\mathcal{P}}| \leq 3K$.  This combined with \Cref{prop-2} completes the proof.
\end{proof}

The key ingredients of the proofs of both Propositions~\ref{proposition:dp step 1} and \ref{prop-2} are two types of deviation inequalities.  
	\begin{itemize}
		\item \textbf{Restricted eigenvalues.}  In the literature on sparse regression, there are several versions of the restricted eigenvalue conditions \citep[see, e.g.][]{buhlmann2011statistics}.  In our analysis, such conditions amount to controlling the probability of the event
			\[
				\mathcal{E}_I = \left\{\sqrt{\sum_{t \in I} \bigl(x_t^{\top}v\bigr)^2} \geq \frac{c_x\sqrt{|I|}}{4} \|v\|_2 - 9 C_x \sqrt{\log(p)}\|v\|_1, \quad \forall v \in \mathbb{R}^p \right\},
			\]
			which is done in \Cref{lem-RE}.
		\item \textbf{Deviations bounds of scaled noise.}  In addition, we need to control the deviations of the quantities of the form
			\begin{equation}\label{eq-lambda-rate-determ}
				\left\|\sum_{t \in I} \varepsilon_t x_t \right\|_{\infty}.
			\end{equation}
			See \Cref{lem-x-bound}.
	\end{itemize}
	
In standard analyses of the performance of the Lasso estimator, as detailed e.g.~in Section 6.2 of \cite{buhlmann2011statistics}, the combination of restricted eigenvalues conditions and large probability bounds on the noise lead to oracle inequalities for the estimation and prediction errors in situations in which there exists no change point and the data are independent. We have extended this line of arguments to the present, more challenging settings, to derive analogous oracle inequalities.  We emphasize a few points in this regard.

	\begin{itemize}
		\item In standard analyses of the Lasso estimator, where there is one and only one true coefficient vector, the magnitude of $\lambda$ is determined as a high-probability upper bound to \eqref{eq-lambda-rate-determ}.  However in our situation, in order to control the $\ell_1$- and $\ell_2$-loss of the estimators $\widehat{\beta}^{\lambda}_I$, where the interval $I$ contains more than one true coefficient vectors, the value of $\lambda$ needs to be inflated by a factor of $\sqrt{d_0}$.  This is detailed in \Cref{lemma:oracle 2}; see, in particular, \eqref{eq-lem10-pf-1}.
		\item The magnitude of the tuning parameter $\gamma$ is determined based on an appropriate oracle inequality for the Lasso and on the number of true change points; more precisely, $\gamma$ can be derived as a high-probability bound for 
			\[
				\left|\sum_{t \in I} \bigl\{(y_t - x_t^{\top}\widehat{\beta}^{\lambda}_I)^2 - (y_t - x_t^{\top}\beta_t^*)^2\bigr\}\right|.
			\]
			See \Cref{lemma:RSS lasso} for details.
			
			The fact that $\gamma$ is linear in the number of change point $K$ is to prompt the consistency.  This is shown in \eqref{eq:Khat=K} in the proof of \Cref{prop-2}. 
		\item The final localization error is obtained by the following calculations.  Assume that there exists one and only one true change point $\eta \in I = (s, e]$.  Define $I_1 = (s, \eta]$ and $I_2 = (\eta_1, e]$.  Let $\beta^*_{I_1}$ and $\beta^*_{I_2}$ be the two true coefficient vectors in $I_1$ and $I_2$, respectively.   For readability, below we will omit all constants here and use the symbol $\lessapprox$ to denote an inequality up to hidden universal constants.  We first assume by contradiction that
			\begin{equation}\label{eq-main-i-cond}
				\min\{|I_1|, |I_2|\} \gtrsim d_0 \log(n \vee p),
 			\end{equation}
 			then use oracle inequalities to establish that
			\begin{align}
				& \sum_{t \in I_1} \{x_t^{\top} (\widehat{\beta}^{\lambda}_I - \beta^*_{I_1})\}^2 + \sum_{t \in I_2} \{x_t^{\top} (\widehat{\beta}^{\lambda}_I - \beta^*_{I_2})\}^2 \nonumber \\
				\lessapprox & \lambda\sqrt{\max\{|I_1|, \, \log(n \vee p)\}} \big\{\sqrt{d_0}\|(\widehat{\beta}^{\lambda}_I - \beta^*_{I_1})(S)\|_2 + \|\widehat{\beta}^{\lambda}_I(S^c)\|_1\big\} \nonumber \\
				& \hspace{2cm} + \lambda\sqrt{\max\{|I_2|, \, \log(n \vee p)\}} \big\{\sqrt{d_0}\|(\widehat{\beta}^{\lambda}_I - \beta^*_{I_2})(S)\|_2 + \|\widehat{\beta}^{\lambda}_I(S^c)\|_1\big\} + \gamma \nonumber\\				
				\lessapprox & \lambda\sqrt{|I_1|} \big\{\sqrt{d_0}\|(\widehat{\beta}^{\lambda}_I - \beta^*_{I_1})(S)\|_2 + \|\widehat{\beta}^{\lambda}_I(S^c)\|_1\big\} \nonumber \\
				& \hspace{2cm} + \lambda\sqrt{|I_2|} \big\{\sqrt{d_0}\|(\widehat{\beta}^{\lambda}_I - \beta^*_{I_2})(S)\|_2 + \|\widehat{\beta}^{\lambda}_I(S^c)\|_1\big\} + \gamma \nonumber\\
				\lessapprox & \frac{\lambda^2 d_0}{c_x^2} + |I_1| \|\widehat{\beta}^{\lambda}_I - \beta^*_{I_1}\|_2^2 + |I_2| \|\widehat{\beta}^{\lambda}_I - \beta^*_{I_2}\|_2^2 + \lambda^2 + (|I_1|^2 + |I_2|^2)\|\widehat{\beta}^{\lambda}_I(S^c)\|_1^2  + \gamma, \label{eq-illustrate-upper}		 
			\end{align}
			where the second inequality follows \eqref{eq-main-i-cond} and the third inequality follows from $2ab \leq a^2 + b^2$ and from setting
			\[
				a = \lambda\sqrt{d_0} \quad \mbox{and} \quad b = \sqrt{|I_1|}\|\widehat{\beta}^{\lambda}_I - \beta^*_{I_1}\|_2.
			\]
			Next we apply the restricted eigenvalue conditions along with standard arguments from the Lasso literature to establish that 
			\begin{align}
				& \sum_{t \in I_1} \{x_t^{\top} (\widehat{\beta}^{\lambda}_I - \beta^*_{I_1})\}^2 + \sum_{t \in I_2} \{x_t^{\top} (\widehat{\beta}^{\lambda}_I - \beta^*_{I_2})\}^2 \nonumber \\
				& \hspace{2cm} \geq c_x^2 |I_1| \|\widehat{\beta}^{\lambda}_I - \beta^*_{I_1}\|^2 + c_x^2 |I_2| \|\widehat{\beta}^{\lambda}_I - \beta^*_{I_2}\|^2 \geq c_x^2 \kappa^2 \epsilon, \label{eq-illustrate-lower}		 
			\end{align}
			where $\epsilon$ is an upper bound on the localization error.  Combining \eqref{eq-illustrate-upper} and \eqref{eq-illustrate-lower} leads to 
			\[
				\epsilon \lesssim \frac{\lambda^2 d_0 + \gamma}{\kappa^2}.
			\]
		\item Finally, the signal-to-noise ratio condition that one needs to assume in order to obtain consistent localization rates is determined by setting $\epsilon \lesssim \Delta$.	
	\end{itemize}



The proofs related with \Cref{algorithm:LR} and \Cref{cor-lr-3} are all based on an oracle inequality of the group Lasso estimator.  Once it is established that
	\begin{equation}\label{eq-intuit-g-lasso}
		\sum_{t = s + 1}^e \|\widehat{\beta}_t - \beta^*_t\|_2^2 \leq \delta \leq \kappa \sqrt{\Delta},
	\end{equation}
	where $\delta \asymp d_0 \log(n \vee p)$ and where there is one and only one change point in the interval $(s, e]$ for both the sequence $\{\widehat{\beta}_t\}$ and $\{\beta^*_t\}$, then the final claim follows immediately that the refined localization error $\epsilon$ satisfies
	\[
		\epsilon \leq \delta/\kappa^2.
	\]
	The group Lasso penalty is deployed to prompt \eqref{eq-intuit-g-lasso} and the designs of the algorithm guarantee the desirability of each working interval.

The proof of \Cref{eq:DP consistent for regression} proceeds through several steps. For convenience, \Cref{fig-rm}  provides a roadmap for the entire proof. Throughout this section, with some abuse of notation, for any interval $I \subset (0, n]$, we denote with $\beta^*_I = |I|^{-1}\sum_{t \in I} \beta^*_t$.

\begin{figure}
\centering
\begin{tikzpicture}[roundnode/.style={circle, draw=Fuchsia!60, fill=Fuchsia!5, very thick, minimum size=7mm}, squarednode/.style={rectangle, draw=CarnationPink!60, fill=CarnationPink!5, very thick, minimum size=7mm}]
\node[squarednode] (1) {\Cref{eq:DP consistent for regression}};
\node[squarednode] (2) [below of= 1, yshift = -1cm, xshift = -2.8cm] {\Cref{proposition:dp step 1}};

\node[squarednode] (18) [right of= 2, xshift = 5cm] {\Cref{prop-2}};

\node[squarednode] (5) [below of= 2, yshift = -1cm, xshift = 0cm] {$\stackrel{\mbox{Case (iii)}}{\mbox{Lemma } \ref{lemma:no change point}}$};
\node[squarednode] (4) [right of= 5, xshift = -4cm] {$\stackrel{\mbox{Case (ii)}}{\mbox{Lemma } \ref{lemma:two change point}}$};
\node[squarednode] (3) [right of= 4, xshift = -4cm] {$\stackrel{\mbox{Case (i)}}{\mbox{Lemma } \ref{lemma:one change point}}$};
\node[squarednode] (6) [right of= 5, xshift = 2cm] {$\stackrel{\mbox{Case (iv)}}{\mbox{Lemma } \ref{lemma:three change point}}$};

\node[squarednode](17) [below of= 18, yshift = -1.2cm] {\Cref{lem-case-5-prop-2-needed}};

\draw[-{Latex[width=2mm]}] (2) -- (1);
\draw[-{Latex[width=2mm]}] (3) -- (2);
\draw[-{Latex[width=2mm]}] (4) -- (2);
\draw[-{Latex[width=2mm]}] (5) -- (2);
\draw[-{Latex[width=2mm]}] (6) -- (2);
\draw[-{Latex[width=2mm]}] (18) -- (1);

\draw[-{Latex[width=2mm]}] (17) -- (18);
\end{tikzpicture}
\caption{Road map to complete  the proof of Theorem~\ref{eq:DP consistent for regression}.   The directed edges mean the heads of the edges are used in the tails of the edges. \label{fig-rm}}
\end{figure}
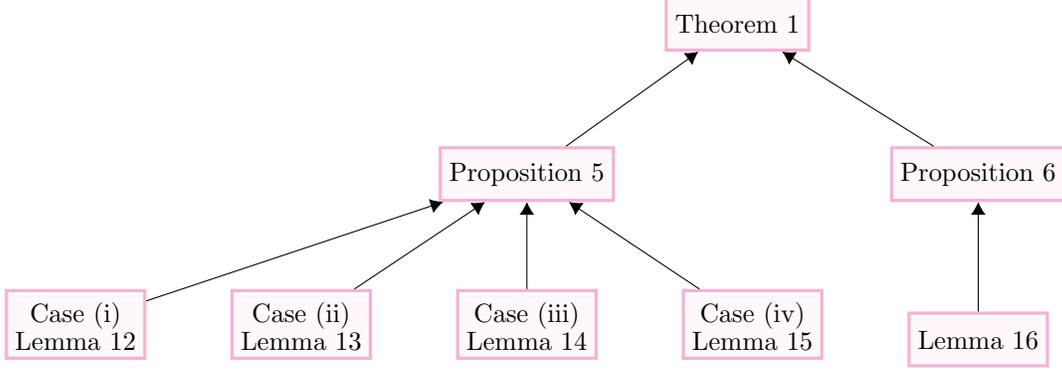
 
\subsection{Large Probability Events}

\begin{lemma}\label{lem-RE}
For \Cref{assume:change point regression model}, under \Cref{assume:high dim coefficient}(\textbf{c}), for any interval $I \subset (0, n]$, it holds that
	\[
		\mathbb{P}\{\mathcal{E}_I\} \geq 1 - c_1 \exp(-c_2 |I|),
	\]
	where $c_1, c_2 > 0$ are absolute constants only depending on the distributions of covariants $\{x_t\}$, and
	\[
		\mathcal{E}_I = \left\{\sqrt{\sum_{t \in I} \bigl(x_t^{\top}v\bigr)^2} \geq \frac{c_x\sqrt{ |I|}}{4} \|v\|_2 - 9 C_x \sqrt{\log(p)}\|v\|_1, \quad v \in \mathbb{R}^p \right\}.
	\]
\end{lemma}

This follows from the same proof as Theorem~1 in \cite{raskutti2010restricted}, therefore we omit the proof of \Cref{lem-RE}.  For interval $I$ satisfying $|I| > C d_0\log(p)$, an immediate consequence of \Cref{lem-RE} is a restricted eigenvalue condition \citep[e.g.][]{van2009conditions, bickel2009simultaneous}.  It will be used repeatedly in the rest of this paper.

It will become clearer in the rest of the paper, we only deal with intervals satisfying $|I| \gtrsim d_0\log(n \vee p)$ when considering the events $\mathcal{E}_I$.

\begin{lemma}\label{lem-x-bound}
For \Cref{assume:change point regression model}, under \Cref{assume:high dim coefficient}(\textbf{c}), for any interval $I \subset (0, n]$, it holds that for any 
	\[
		\lambda \geq \lambda_1 := C_{\lambda}\sigma_{\varepsilon} \sqrt{\log(n \vee p)},
	\]
	where $C_{\lambda} > 0$ is a large enough absolute constant such that, we have
	\[
		\mathbb{P}\{\mathcal{B}_I(\lambda)\} > 1 - 2(n \vee p)^{-c_3},
	\]	
	where
	\[
		\mathcal{B}_I(\lambda) = \left\{\left\|\sum_{t \in I} \varepsilon_t x_t \right\|_{\infty} \leq \lambda \sqrt{\max\{|I|, \, \log(n \vee p)\}}/8 \right\},
	\]
	where $c_3 > 0$ is an absolute constant depending only on the distributions of covariants $\{x_t\}$ and $\{\varepsilon_t\}$.
\end{lemma}

For notational simplicity, we drop the dependence on $\lambda$ in the notation $\mathcal{B}_I(\lambda)$.

\begin{proof}
	Since $\varepsilon_t$'s are sub-Gaussian random variables and $x_t$'s are sub-Gaussian random vectors, we have that $\varepsilon_t x_t$'s are sub-Exponential random vectors with parameter $C_x \sigma_{\varepsilon}$ \citep[see e.g.~Lemma~2.7.7 in][]{vershynin2018high}.  It then follows from Bernstein's inequality \citep[see e.g.~Theorem~2.8.1 in][]{vershynin2018high} that for any $t > 0$,
		\[
			\mathbb{P}\left\{\left\|\sum_{t \in I} \varepsilon_t x_t \right\|_{\infty} > t\right\} \leq 2p \exp\left\{-c \min\left\{\frac{t^2}{|I|C_x^2\sigma^2_{\varepsilon}}, \, \frac{t}{C_x \sigma_{\varepsilon}}\right\}\right\}.
		\]
		Taking 
		\[
			t = C_{\lambda}C_x/4 \sigma_{\varepsilon} \sqrt{\log(n \vee p)} \sqrt{\max\{|I|, \, \log(n \vee p)\}}
		\] 
		yields that 
		\[
			\mathbb{P}\{\mathcal{B}_I\} > 1 - 2(n \vee p)^{-c_3},
		\]
		where $c_3 > 0$ is an absolute constant depending on $C_{\lambda}, C_x, \sigma_{\varepsilon}$.
\end{proof}

\subsection{Auxiliary Lemmas}

\begin{lemma} \label{lemma:lasso}
For \Cref{assume:change point regression model}, under \Cref{assume:high dim coefficient}(\textbf{a}) and (\textbf{c}), if there exists no true change point in $I = (s, e]$, with $|I| > 288^2 C^2_x d_0 \log(n \vee p)/c_x^2$ and
	\[
		\lambda \geq \lambda_1 := C_{\lambda}\sigma_{\varepsilon} \sqrt{\log(n \vee p)},
	\] 
	where $C_{\lambda} >  0$ being an absolute constant, it holds that
	\begin{align*}
		& \mathbb{P}\left\{\bigl\|\widehat{\beta}^{\lambda}_I - \beta^*_I\bigr\|_2 \leq \frac{C_3 \lambda \sqrt{d_0}} {\sqrt{|I|}}, \quad \bigl\|\widehat{\beta}^{\lambda}_I - \beta^*_I\bigr\|_1 \leq \frac{C_3 \lambda d_0}{\sqrt{|I|}} \right\} \\
		& \hspace{5cm} \geq 	1 - c_1 (n \vee p)^{-288^2C_x^2d_0c_2/c_x^2} - 2(n \vee p)^{-c_3},
	\end{align*}
	where $C_3 > 0$ is an absolute constant depending on all the other absolute constants, $c_1, c_2, c_3$ are absolute constants defined in Lemmas~\ref{lem-RE} and \ref{lem-x-bound}.
\end{lemma}

\begin{proof}
Let $v = \widehat{\beta}^{\lambda}_I - \beta^*_I$.  Since $|I| > \log(n \vee p)$, it follows from the definition of $\widehat{\beta}^{\lambda}_I$ that
	\[
		\sum_{t \in I} (y_t - x_t^{\top}\widehat{\beta}^{\lambda}_I)^2 + \lambda \sqrt{|I|}\|\widehat{\beta}^{\lambda}_I\|_1 \leq \sum_{t \in I} (y_t - x_t^{\top}\beta^*_I)^2 + \lambda \sqrt{|I|}\|\beta^*_I\|_1,
	\]
	which leads to
	\begin{equation}\label{eq-lem9-pf-1}
		\sum_{t \in I}(x_t^{\top}v)^2 + \lambda\sqrt{|I|}\|\widehat{\beta}^{\lambda}_I\|_1 \leq \lambda \sqrt{|I|}\|\beta^*_I\|_1 + 2 \sum_{t \in I} \varepsilon_t x_t^{\top}v \leq \lambda \sqrt{|I|}\|\beta^*_I\|_1 +  \frac{\lambda}{2}\sqrt{|I|}\|v\|_1,
	\end{equation}
	where the last inequality holds on the event $\mathcal{B}_I$, with the choice of $\lambda$ and due to \Cref{lem-x-bound}.  Note that
	\begin{equation}\label{eq-lem9-pf-2}
		\|\widehat{\beta}^{\lambda}_I\|_1 \geq \|\beta^*_I(S)\|_1 - \|v(S)\|_1 + \|\widehat{\beta}^{\lambda}_I(S^c)\|_1
	\end{equation}
	and 
	\begin{equation}\label{eq-lem9-pf-3}
		\|v\|_1 = \|v(S)\|_1 + \|\widehat{\beta}^{\lambda}_I(S^c)\|_1.
	\end{equation}
	Combining \eqref{eq-lem9-pf-1}, \eqref{eq-lem9-pf-2} and \eqref{eq-lem9-pf-3} yields
	\begin{equation}\label{eq-lem9-pf-4}
		\sum_{t \in I}(x_t^{\top}v)^2 + \frac{\lambda}{2}\sqrt{|I|}\|\widehat{\beta}^{\lambda}_I(S^c)\|_1 \leq \frac{3\lambda}{2}\sqrt{|I|}\|\widehat{\beta}^{\lambda}_I(S)\|_1,
	\end{equation}
	which in turn implies
	\[
		\|\widehat{\beta}^{\lambda}_I(S^c)\|_1 \leq 3\|\widehat{\beta}^{\lambda}_I(S)\|_1.
	\]
	
On the event of $\mathcal{E}_I$, it holds that
	\begin{align}
		\sqrt{\sum_{t \in I} \bigl(x_t^{\top}v\bigr)^2} & \geq \frac{c_x\sqrt{|I|}}{4} \|v\|_2 - 9 C_x \sqrt{\log(p)}\|v\|_1 \nonumber \\
		& = \frac{c_x\sqrt{|I|}}{4} \|v\|_2 - 9 C_x \sqrt{\log(p)}\|v(S)\|_1 - 9 C_x \sqrt{\log(p)}\|v(S^c)\|_1 \nonumber \\
		& \geq \frac{c_x\sqrt{|I|}}{4} \|v\|_2 - 36 C_x \sqrt{\log(p)}\|v(S)\|_1 \geq  \frac{c_x\sqrt{|I|}}{4} \|v\|_2 - 36 C_x \sqrt{d_0 \log(p)}\|v(S)\|_2 \nonumber \\
		& \geq \left(\frac{c_x\sqrt{|I|}}{4} - 36 C_x \sqrt{d_0 \log(p)}\right) \|v\|_2 > \frac{c_x\sqrt{|I|}}{8} \|v\|_2, \label{eq-lem9-pf-5}
	\end{align}
	where the second inequality follows from \eqref{eq-lem9-pf-4}, the third inequality follows from \Cref{assume:high dim coefficient}(\textbf{a}) and the last inequality follows from the choice of $|I|$.
	
Combining \eqref{eq-lem9-pf-4} and \eqref{eq-lem9-pf-5} leads to
	\[
		\frac{c_x^2|I|}{64} \|v\|_2^2 \leq \frac{3\lambda}{2}\sqrt{|I|} \|v(S)\|_1 \leq \frac{3\lambda}{2}\sqrt{|I|d_0} \|v\|_2,
	\]	
	therefore
	\[
		\|v\|_2 \leq \frac{96 \lambda\sqrt{d_0}}{\sqrt{|I|} c_x^2}
	\]
	and 
	\[
		\|v\|_1 = \|v(S)\|_1 + \|v(S^c)\|_1 \leq 4\|v(S)\|_1 \leq 4\sqrt{d_0}\|v\|_2 \leq \frac{384 \lambda d_0}{\sqrt{|I|} c_x^2}.
	\]
\end{proof}

\begin{lemma}\label{lemma:RSS lasso} 
For \Cref{assume:change point regression model}, under \Cref{assume:high dim coefficient}(\textbf{a}) and (\textbf{c}), if there exists no true change point in $I = (s, e]$, and
	\[
		\lambda \geq \lambda_1 : = C_{\lambda}\sigma_{\varepsilon} \sqrt{\log(n \vee p)},
	\]
	where $C_{\lambda} > 0$ being an absolute constant, it holds that if $|I| \geq 288^2 C^2_x d_0 \log(n \vee p)/c_x^2$, then
	\begin{align*}
		& \mathbb{P}\left\{\left|\sum_{t \in I} \bigl\{(y_t - x_t^{\top}\widehat{\beta})^2 - (y_t - x_t^{\top}\beta^*)^2\bigr\}\right| \leq \lambda^2 d_0 \right\} \\
		& \hspace{2cm}\geq 	1 - c_1 (n \vee p)^{-288^2C_x^2d_0c_2/c_x^2} - 2(n \vee p)^{-c_3};
	\end{align*}
	if $|I| < 288^2 C^2_x d_0 \log(n \vee p)/c_x^2$, then 
	\begin{align*}
		\mathbb{P}\left\{\left|\sum_{t \in I} \bigl\{(y_t - x_t^{\top}\widehat{\beta})^2 - (y_t - x_t^{\top}\beta^*)^2\bigr\}\right| \leq C_4 \lambda \sqrt{\log(n \vee p)} d_0^{3/2}\right\} \geq 1 - 2(n \vee p)^{-c_3},
	\end{align*}
	where $C_4 > 0$ is an absolute constant depending on all the other constants.
\end{lemma} 

\begin{proof}
To ease notation, in this proof, let $\widehat{\beta} = \widehat{\beta}^\lambda_{I}$ and $\beta^* = \beta^*_I$. 

\vskip 3mm
\noindent \textbf{Case 1.}  If $|I| \geq 288^2 C^2_x d_0 \log(n \vee p)/c_x^2$, then $|I| > \log(n \vee p)$.  With probability at least $1 - c_1 \exp(-c_2 |I|) - 2(n \vee p)^{-c_3}$, we have that
	\begin{align*}
		& \sum_{t \in I} \bigl\{(y_t - x_t^{\top}\widehat{\beta})^2 - (y_t - x_t^{\top}\beta^*)^2\bigr\} \leq \lambda \sqrt{|I|} \|\beta^*\|_1 - \lambda \sqrt{|I|} \|\widehat{\beta}\|_1\leq \lambda \sqrt{|I|} \|\widehat{\beta} - \beta^*\|_1 \leq C_3 \lambda^2 d_0,
	\end{align*}
	where the fist inequality follows from the definition of $\widehat{\beta}$ and the second is due to \Cref{lemma:lasso}.
		
\vskip 3mm
\noindent \textbf{Case 2.}  If $|I| < 288^2 C^2_x d_0 \log(n \vee p)/c_x^2$, then
	\begin{align*}
		\sum_{t \in I} \bigl\{(y_t - x_t^{\top}\widehat{\beta})^2 - (y_t - x_t^{\top}\beta^*)^2\bigr\}	 & \leq \lambda \sqrt{\max\{|I|, \, \log(n \vee p)\}}\|\beta^*\|_1 \leq C_4 \lambda \sqrt{\log(n \vee p)} d_0^{3/2},
	\end{align*}
	since $\|\beta^*\|_1 \leq C_{\beta}d_0$.  In addition, it holds with probability at least $1 - 2(n \vee p)^{-c_3}$ that
	\begin{align*}
		 & \sum_{t \in I} \bigl\{(y_t - x_t^{\top}\beta^*)^2 - (y_t - x_t^{\top}\widehat{\beta})^2\bigr\} = - \sum_{t \in I} (x_t^{\top} \beta^* - x_t^{\top} \widehat{\beta})^2 + 2\sum_{t \in I} \varepsilon_t x_t^{\top}(\widehat{\beta} - \beta^*) \\
		 \leq & - \sum_{t \in I} (x_t^{\top} \beta^* - x_t^{\top} \widehat{\beta})^2 + \sum_{t \in I} (x_t^{\top} \beta^* - x_t^{\top} \widehat{\beta})^2 + \sum_{t \in I} \varepsilon_t^2 \leq \sum_{t \in I} \varepsilon_t^2   \\
		 \leq & \max\{\sqrt{|I|\log(n \vee p)}, \, \log(n \vee p)\} \leq C_4 \lambda \sqrt{\log(n \vee p)} d_0^{3/2},
	\end{align*}
	where the first inequality follow from $2ab \leq a^2 + b^2$ and letting $a = \varepsilon_t$, $b = x_t^{\top}(\widehat{\beta} - \beta^*)$,  the third inequality follows from the sub-Gaussianity of $\{\varepsilon_t\}$.
\end{proof}
 
\begin{lemma} \label{lemma:oracle 2}
For \Cref{assume:change point regression model}, under \Cref{assume:high dim coefficient}(\textbf{a})-(\textbf{c}), for any interval $I = (s, e]$ and
	\[
		\lambda \geq \lambda_2 : = C_{\lambda}\sigma_{\varepsilon} \sqrt{d_0 \log(n \vee p)}, 
	\] 
	where $C_{\lambda} > 8C_{\beta}C_x/\sigma_{\varepsilon}$, it holds with probability at least of $1 - 2(n \vee p)^{-c}$ that,
	\begin{align*}	
		\| \widehat \beta^\lambda_{I} (S^c)\|_1 \le 3 \| \widehat \beta^\lambda_{I} (S )\|_1.
	\end{align*}

If in addition, the interval $I$ satisfies $|I| > 288^2 C^2_x d_0 \log(n \vee p)/c_x^2$, it holds with probability at least $1 - c_1 (n \vee p)^{-288^2C_x^2d_0c_2/c_x^2} - 2(n \vee p)^{-c_3}$ that
	\begin{align*}
		\left \|\widehat\beta_{I}^ \lambda - \frac{1}{|I|} \sum_{t \in I} \beta_t^* \right\|_2 \leq  \frac{C_5 \lambda \sqrt{d_0}}{\sqrt{|I|}} \quad \mbox{and} \quad \left \|  \widehat\beta_{I}^ \lambda - \frac{1}{|I|} \sum_{t \in I} \beta_t^* \right\|_1 \le  \frac{C_5 \lambda d_0}{\sqrt{|I|}},
	\end{align*}
	where $C_5 > 0$ is an absolute constant depending on other constants.
\end{lemma}
 
\begin{proof} 
Denote $\widehat{\beta} = \widehat\beta_{I}^\lambda$ and $\beta^* = (|I|)^{-1}\sum_{t \in I} \beta^*_t$.  It follows from the definition of $\widehat{\beta}$ that 
	\[
		\sum_{t \in I} (y_t - x_t^{\top}\widehat{\beta})^2 + \lambda \sqrt{\max\{|I|, \, \log(n \vee p)\}} \bigl\|\widehat{\beta}\bigr\|_1 \leq \sum_{t \in I} (y_t - x_t^{\top}\beta^*)^2 + \lambda \sqrt{\max\{|I|, \, \log(n \vee p)\}} \bigl\|\beta^*\bigr\|_1,
	\]
	which leads to
	\begin{align*}
		\sum_{t \in I} \bigl\{x_t^{\top}(\widehat{\beta} - \beta^*)\bigr\}^2 +  2\sum_{t \in I}(y_t - x_t^{\top}\beta^*)x_t^{\top}(\beta^* - \widehat{\beta})  + \lambda \sqrt{\max\{|I|, \, \log(n \vee p)\}} \bigl\|\widehat{\beta}\bigr\|_1 \\
		\leq \lambda \sqrt{\max\{|I|, \, \log(n \vee p)\}} \bigl\|\beta^*\bigr\|_1,
	\end{align*}
	therefore
	\begin{align}
		& \sum_{t \in I} \bigl\{x_t^{\top}(\widehat{\beta} - \beta^*)\bigr\}^2 + 2(\widehat{\beta} - \beta^*)^{\top}\sum_{t \in I} x_t x_t^{\top}(\beta^* - \beta^*_t) \nonumber \\
		& \hspace{3cm} \leq 2\sum_{t \in I} \varepsilon_t x_t^{\top}(\widehat{\beta} - \beta^*) + \lambda \sqrt{\max\{|I|, \, \log(n \vee p)\}}\bigl(\bigl\|\beta^*\bigr\|_1 - \bigl\|\widehat{\beta}\bigr\|_1\bigr).	 \label{eq-lem10-pf-2}
	\end{align}
	
We bound	
	\[
		\left\|\sum_{t \in I} x_t x_t^\top (\beta^*  -\beta^*_t) \right\|_{\infty}.  
	\]
	For any $k \in \{1, \ldots, p\}$, the $k$th entry of $\sum_{t \in I} x_t x_t^\top (\beta^*  -\beta^*_t)$ satisfies that
	\begin{align*}
		& \mathbb{E}\left\{\sum_{t \in I} \sum_{j = 1}^p x_t(k)x_t(j)(\beta^*(j) - \beta^*_t(j))\right\} = \sum_{t \in I}\sum_{j = 1}^p \mathbb{E}\{x_t(k) x_t(j)\}\{\beta^*(j) - \beta^*_t(j)\} \\
		= & \sum_{j = 1}^p \mathbb{E}\{x_1(k) x_1(j)\} \sum_{t \in I}\{\beta^*(j) - \beta^*_t(j)\} = 0.
	\end{align*}

Note that $x_t^{\top}(\beta^* - \beta^*_t)$'s are sub-Gaussian random variables with a common parameter $2C_{\beta}C_x\sqrt{d_0}$, and $x_t$'s are sub-Gaussian random vectors with parameter $C_x$.  Therefore due to sub-Exponential inequalities \citep[e.g. Proposition~2.7.1 in][]{vershynin2018high}, it holds with probability at least of $1 - 2(n \vee p)^{-c}$ that,
	\begin{align}
		\left\|\sum_{t \in I} x_t x_t^\top (\beta^*  -\beta^*_t) \right\|_{\infty} \leq 2C_x C_{\beta}\sqrt{d_0}\max\{\sqrt{|I|\log(n \vee p)}, \, \log(n \vee p)\}\nonumber \\
		\leq \lambda \sqrt{\max\{|I|, \, \log(n \vee p)\}}/4.\label{eq-lem10-pf-1}
	\end{align}
	On the event $\mathcal{B}_I$, combining \eqref{eq-lem10-pf-2} and \eqref{eq-lem10-pf-1} yields
	\begin{align*}
		&\sum_{t \in I} \bigl\{x_t^{\top}(\widehat{\beta} - \beta^*)\bigr\}^2 + \lambda \sqrt{\max\{|I|, \, \log(n \vee p)\}} \bigl\|\widehat{\beta}\bigr\|_1 \\
		&\hspace{2cm}\leq \lambda/2 \sqrt{\max\{|I|, \, \log(n \vee p)\}} \bigl\|\beta^* - \widehat{\beta}\bigr\|_1 + \lambda \sqrt{\max\{|I|, \, \log(n \vee p)\}} \bigl\|\beta^*\bigr\|_1. 
	\end{align*}
	The final claims follow from the same arguments as in \Cref{lemma:lasso}.
\end{proof}

\subsection{All cases in Proposition~\ref{proposition:dp step 1}}

\begin{lemma}[Case (i)]\label{lemma:one change point} 
With the conditions and notation in \Cref{proposition:dp step 1}, assume that $I = (s, e] \in \widehat{\mathcal{P}}$ has one and only one true change point $\eta$.  Denote $I_1 = (s, \eta]$, $I_2 = (\eta, e]$ and $\|\beta_{I_1}^* - \beta_{I_2}^*\|_2 = \kappa$.  If, in addition, it holds that  
	\begin{align}\label{eq:one change point} 
		\sum_{t \in I}(y_t - x_t^{\top}\widehat{\beta}^\lambda_I)^2 \leq \sum_{t \in I_1}(y_t - x_t^{\top}\widehat{\beta}^\lambda_{I_1})^2 + \sum_{t \in I_2}(y_t - x_t^{\top}\widehat{\beta}^\lambda_{I_2})^2 + \gamma,
	\end{align}
	then with 
	\[
		\lambda \geq \lambda_2 = C_{\lambda}\sigma_{\varepsilon} \sqrt{d_0 \log(n \vee p)}, 
	\]
	where $C_{\lambda} > 8C_{\beta}C_x/\sigma_{\varepsilon}$, it holds with probability at least $1 - 2c_1 (n \vee p)^{-288^2C_x^2d_0c_2/c_x^2} - 2(n \vee p)^{-c_3}$ that, that
	\[
		\min\{|I_1 |,\, |I_2|\} \leq C_{\epsilon} \left(\frac{\lambda^2 d_0 + \gamma}{\kappa^2}\right).
	\]
\end{lemma} 

\begin{proof} 
First we notice that with the choice of $\lambda$, it holds that 
	\[
		\lambda \geq \max\{\lambda_1, \, \lambda_2\},
	\]
	and therefore we can apply Lemmas~\ref{lemma:lasso}, \ref{lemma:RSS lasso} and \ref{lemma:oracle 2} when needed.

We prove by contradiction, assuming that
	\begin{equation}\label{eq-contra}
		\min\{|I_1|, \, |I_2|\} > C_{\epsilon} \left(\frac{\lambda^2 d_0 + \gamma}{\kappa^2}\right) > 288^2 C^2_x d_0 \log(n \vee p)/c_x^2,
	\end{equation}
	where the second inequality follows from the observation that $\kappa^2 \leq 4 d_0 C^2_{\beta}$.  Therefore we also have
	\[
		\min\{|I_1|, \, |I_2|\} > \log(n \vee p).
	\]

It follows from \Cref{lemma:RSS lasso} and \eqref{eq:one change point} that, with probability at least $1 - 2c_1 (n \vee p)^{-288^2C_x^2d_0c_2/c_x^2} - 2(n \vee p)^{-c_3}$ that, that
	\begin{align}
		& \sum_{t \in I_1}(y_t - x_t^{\top}\widehat{\beta}^\lambda_I)^2 + \sum_{t \in I_2}(y_t - x_t^{\top}\widehat{\beta}^\lambda_I)^2 = \sum_{ t\in I}(y_t - x_t^{\top}\widehat{\beta}^\lambda_I)^2  \nonumber \\
		\leq & \sum_{t \in I_1}(y_t - x_t^{\top}\widehat{\beta}^\lambda_{I_1})^2 + \sum_{t \in I_2}(y_t - x_t^{\top}\widehat{\beta}^\lambda_{I_2})^2 + \gamma \nonumber \\
		\leq & \sum_{t \in I_1}(y_t - x_t^{\top}\beta^*_{I_1})^2 + \sum_{t \in I_2}(y_t - x_t^{\top}\beta^*_{I_2})^2 + \gamma + 2C_3 \lambda^2 d_0. \label{eq-lem-1cp-pf-1}
	\end{align}

Denoting $\Delta_i = \widehat{\beta}^\lambda_{I} - \beta_{I_i}^*$, $i = 1, 2$, \eqref{eq-lem-1cp-pf-1} leads to that  
	\begin{align}
		& \sum_{t\in I_1} (x_t^{\top}\Delta_1)^2 + \sum_{t\in I _2 } (x_t^{\top}\Delta_2)^2 \leq 2\sum_{t\in I_1} \varepsilon_t x_t^{\top} \Delta_1 + 2\sum_{t\in I_2} \varepsilon_t x_t^{\top} \Delta_2  + \gamma + 2C_3 \lambda^2 d_0 \nonumber \\
		\leq & 2 \left\|\sum_{t\in I_1} \varepsilon_t x_t  \right\|_{\infty} \|\Delta_1\|_1 + 2 \left\|\sum_{t\in I_2} \varepsilon_t x_t  \right\|_{\infty} \|\Delta_2\|_1 + \gamma + 2C_3 \lambda^2 d_0 \nonumber \\
		\leq & 2 \left\|\sum_{t\in I_1} \varepsilon_t x_t  \right\|_{\infty} \bigl(\|\Delta_1(S)\|_1 + \|\Delta_1(S^c)\|_1\bigr) + 2 \left\|\sum_{t\in I_2} \varepsilon_t x_t  \right\|_{\infty} \bigl(\|\Delta_2(S)\|_1 + \|\Delta_2(S^c)\|_1\bigr) \nonumber \\
		& \hspace{5cm} + \gamma + 2C_3 \lambda^2 d_0 \nonumber \\
		\leq & 2 \left\|\sum_{t\in I_1} \varepsilon_t x_t  \right\|_{\infty} \bigl(\sqrt{d_0}\|\Delta_1(S)\|_2 + \|\Delta_1(S^c)\|_1\bigr) + 2 \left\|\sum_{t\in I_2} \varepsilon_t x_t  \right\|_{\infty} \bigl(\sqrt{d_0}\|\Delta_2(S)\|_2 + \|\Delta_2(S^c)\|_1\bigr) \nonumber \\
		& \hspace{5cm} + \gamma + 2C_3 \lambda^2 d_0.\label{eq:one change point first step}
	\end{align}
	
On the events $\mathcal{B}_{I_1} \cap \mathcal{B}_{I_2}$, it holds that
	\begin{align}
		\eqref{eq:one change point first step} & \leq \lambda/2 \big(\sqrt{|I_1|d_0}\|\Delta_1(S)\|_2 + \sqrt{|I_1|}\|\Delta_1(S^c)\|_1 + \sqrt{|I_2|d_0} \|\Delta_2(S)\|_2 \nonumber \\
		& \hspace{1cm} + \sqrt{|I_2|}\|\Delta_2(S^c)\|_1\big) + \gamma + 2C_3 \lambda^2 d_0 \nonumber \\
		& \leq \frac{32 \lambda^2 d_0}{c_x^2} + \frac{c_x^2 |I_1| \|\Delta_1\|_2^2}{256} + \frac{c_x^2 |I_2| \|\Delta_2\|_2^2}{256} + \frac{\lambda (\sqrt{|I_1|} + \sqrt{|I_2|})}{2} \|\widehat{\beta}^{\lambda}_I(S^c)\|_1 + \gamma + 2C_3 \lambda^2 d_0 \nonumber \\
		& \leq \frac{32 \lambda^2 d_0}{c_x^2} + \frac{c_x^2 |I_1| \|\Delta_1\|_2^2}{256} + \frac{c_x^2 |I_2| \|\Delta_2\|_2^2}{256} + \gamma + 4C_3 \lambda^2 d_0, \label{eq-upper-long}
	\end{align}
	where the second inequality follows from $2ab \leq a^2 + b^2$, letting
	\[
		a = 4\lambda \sqrt{d_0}/c_x \quad \mbox{and} \quad b = c_x \sqrt{|I_j|}\|\Delta_1\|_2/16, \quad j = 1, 2,
	\]
	and the last inequality follows from \Cref{lemma:oracle 2}.

\vskip 3mm
Note that 
	\begin{align*}
		\|\Delta_1\|_1 \leq \|\Delta_1(S)\|_1 + \|\Delta_1(S^c)\|_1 \leq \sqrt{d_0}\|\Delta_1\|_2 + \frac{C_5 \lambda d_0 }{\sqrt{|I_1|}},
	\end{align*}
	which combines with \eqref{eq-contra}, on the event $\mathcal{E}_{I_1}$, leads to
	\[
		\sqrt{\sum_{t \in I_1} (x_t^{\top}\Delta_1)^2} > \frac{c_x \sqrt{|I_1|}}{4} \|\Delta_1\|_2 - 9 C_x \sqrt{\log(p)}\|\Delta_1\|_1 > \frac{c_x \sqrt{|I_1|}}{8} \|\Delta_1\|_2 - \frac{9C_5 C_x\lambda d_0 \sqrt{\log(p)} }{c_x^2 \sqrt{|I_1|}}.
	\]

Moreover, we have
	\begin{align}
		& \sqrt{|I_1|} \|\Delta_1\|_2 + \sqrt{|I_2|} \|\Delta_2\|_2 \geq \sqrt{|I_1| \|\Delta_1\|_2^2 + |I_2| \|\Delta_2\|_2^2} \nonumber \\
		\geq & \sqrt{\inf_{v \in \mathbb{R}^p} \{|I_1| \|\beta^*_{\eta} - v\|^2 + |I_2| \|\beta^*_{\eta + 1} - v\|^2\}} = \kappa \sqrt{\frac{|I_1||I_2|}{|I|}} \geq \frac{\kappa}{\sqrt{2}}\min\{\sqrt{|I_1|}, \sqrt{|I_2|} \}. \label{eq-lower-long}
	\end{align}
	
Therefore, on the event $\mathcal{E}_{I_1} \cap \mathcal{E}_{I_2} \cap \mathcal{B}_{I_1} \cap \mathcal{B}_{I_2}$, combining \eqref{eq:one change point first step} and \eqref{eq-upper-long},  we have that
	\begin{align*}
		&  \sqrt{|I_1|} \|\Delta_1\|_2 + \sqrt{|I_2|} \|\Delta_2\|_2 \leq \frac{8}{c_x} \left(\sqrt{\sum_{t \in I_1} (x_t^{\top}\Delta_1)^2} + \sqrt{\sum_{t \in I_2} (x_t^{\top}\Delta_2)^2} \right) \\
		& \hspace{5cm} + \frac{8}{c_x}\left(\frac{9C_5 C_x\lambda d_0 \sqrt{\log(p)} }{c_x^2 \sqrt{|I_1|}} + \frac{9C_5 C_x\lambda d_0 \sqrt{\log(p)} }{c_x^2 \sqrt{|I_2|}}\right) \\
		& \leq \frac{8\sqrt{2}}{c_x} \sqrt{\frac{32 \lambda^2 d_0}{c_x^2} + \frac{c_x^2 |I_1| \|\Delta_1\|_2^2}{256} + \frac{c_x^2 |I_2| \|\Delta_2\|_2^2}{256} + \gamma + 4C_3 \lambda^2 d_0} \\
		& \hspace{5cm} + \frac{8}{c_x}\left(\frac{9C_5 C_x\lambda d_0 \sqrt{\log(p)} }{c_x^2 \sqrt{|I_1|}} + \frac{9C_5 C_x\lambda d_0 \sqrt{\log(p)} }{c_x^2 \sqrt{|I_2|}}\right)  \\
		& \leq \frac{64\lambda \sqrt{d_0}}{c_x^2} + \frac{\sqrt{2} \sqrt{|I_1|}\|\Delta_1\|_2}{2} + \frac{\sqrt{2} \sqrt{|I_2|}\|\Delta_2\|_2}{2} + 	\frac{8\sqrt{2\gamma}}{c_x}	+ \frac{16\sqrt{2C_3} \lambda \sqrt{d_0}}{c_x} + \frac{C_5\lambda \sqrt{d_0}}{2c_x^2},
	\end{align*}
	which implies that
	\begin{align}\label{eq-2-sqrt-2}
		\frac{2 - \sqrt{2}}{2} \left(\sqrt{|I_1|} \|\Delta_1\|_2 + \sqrt{|I_2|} \|\Delta_2\|_2 \right) \leq \frac{128 + 32\sqrt{2}c_x\sqrt{C_3} + C_5}{2c_x^2} \lambda\sqrt{d_0} + 	\frac{8\sqrt{2\gamma}}{c_x}.
	\end{align}
	Combining \eqref{eq-lower-long} and \eqref{eq-2-sqrt-2} yields
	\begin{align*}
		\frac{2 - \sqrt{2}}{2\sqrt{2}}\kappa \sqrt{\min\{|I_1|, \, |I_2|\}} \leq \frac{128 + 32\sqrt{2}c_x\sqrt{C_3} + C_5}{2c_x^2} \lambda\sqrt{d_0} + 	\frac{8\sqrt{2\gamma}}{c_x},	 
	\end{align*}
	therefore
	\[
		\min\{|I_1|, \, |I_2|\} \leq C_{\epsilon} \left(\frac{\lambda^2 d_0 + \gamma}{\kappa^2}\right),
	\]
	which is a contradiction with \eqref{eq-contra}.
\end{proof}

\begin{lemma}[Case (ii)] \label{lemma:two change point}
For \Cref{assume:change point regression model}, under \Cref{assume:high dim coefficient}, with 
	\[
		\lambda \geq \lambda_2 = C_{\lambda}\sigma_{\varepsilon} \sqrt{d_0 \log(n \vee p)}, 
	\]
	where $C_{\lambda} > 8C_{\beta}C_x/\sigma_{\varepsilon}$, $I = (s, e]$ containing exactly two change points $\eta_1$ and $\eta_2$.  Denote $ I_1 =(s,\eta_1] $, $I_2 = (\eta_1,\eta_2] $, $I_3 =(\eta_2,e]$, $\|  \beta_{I_1} ^*  -\beta_{I_2} ^* \|_2 = \kappa_1$ and $\|  \beta_{I_2} ^*  -\beta_{I_3} ^* \|_2 = \kappa_2$.  If in addition it holds that  
	\[
		 \sum_{t\in I } (y_t -x_t^{\top} \widehat \beta^\lambda_ I )^2 \leq \sum_{t\in I_1 } (y_t -x_t^{\top} \widehat \beta^\lambda_ {I_1}  )^2 + \sum_{t\in I_2 } (y_t -x_t^{\top} \widehat \beta^\lambda_ {I_2}  )^2  + \sum_{t\in I_3} (y_t -x_t^{\top} \widehat \beta^\lambda_ {I_3}  )^2 + 2\gamma, 
	\]
	then 
	\begin{align*}
		\max\{|I_1|, \, |I_3|\} \leq C_{\epsilon} \left(\frac{\lambda^2 d_0 + \gamma}{\kappa^2}\right),
	\end{align*} 
	with probability at least $1 - 3c_1 (n \vee p)^{-288^2C_x^2d_0c_2/c_x^2} - 2(n \vee p)^{-c_3}$.
\end{lemma} 

\begin{proof} 
First we notice that with the choice of $\lambda$, it holds that 
	\[
		\lambda \geq \max\{\lambda_1, \, \lambda_2\},
	\]
	and therefore we can apply Lemmas~\ref{lemma:lasso}, \ref{lemma:RSS lasso} and \ref{lemma:oracle 2} when needed.

By symmetry, it suffices to show that
	\[
		|I_1| \leq C_{\epsilon} \left(\frac{\lambda^2 d_0 + \gamma}{\kappa^2}\right).
	\] 
	We prove by contradiction, assuming that
	\begin{equation}\label{eq-contra-2}
		|I_1| > C_{\epsilon} \left(\frac{\lambda^2 d_0 + \gamma}{\kappa^2} \right) > 288^2 C^2_x d_0 \log(n \vee p)/c_x^2,
	\end{equation}
	where the second inequality follows from the observation that $\kappa^2 \leq 4 d_0 C^2_{\beta}$.  Therefore we have $|I_1| > \log(n \vee p)$.  Denote $\Delta_i = \widehat \beta^\lambda_I  - \beta_{I_i}^*$, $i = 1, 2, 3$.  We then consider the following two cases.

\vskip 3mm
\noindent {\bf Case 1.}  If 
	\[
		|I_3| > 288^2 C^2_x d_0 \log(n \vee p)/c_x^2,
	\]
	then $|I_3| > \log(n \vee p)$.  It follows from \Cref{lemma:RSS lasso} that the following holds with probability at least $1 - 3c_1 (n \vee p)^{-288^2C_x^2d_0c_2/c_x^2} - 2(n \vee p)^{-c_3}$ that,
	\begin{align}
		& \sum_{t\in I } (y_t -x_t^{\top} \widehat \beta^\lambda_ I )^2 \leq \sum_{t\in I_1 } (y_t -x_t^{\top} \widehat \beta^\lambda_ {I_1}  )^2 + \sum_{t\in I_2 } (y_t -x_t^{\top} \widehat \beta^\lambda_ {I_2}  )^2  + \sum_{t\in I_3} (y_t -x_t^{\top} \widehat \beta^\lambda_ {I_3}  )^2 + 2\gamma \nonumber \\
		\leq &  \sum_{t \in I_1 } (y_t -x_t^{\top}   \beta^* _ {I_1}  )^2 + \sum_{t \in I_2 } (y_t -x_t^{\top}   \beta^* _ {I_2}  )^2 + \sum_{t \in I_3 } (y_t -x_t^{\top}   \beta^* _ {I_3}  )^2 + 3C_3 \lambda^2 d_0 + 2\gamma,  \nonumber
	\end{align}
	which implies that 
	\begin{align}
	& \sum_{i=1}^3 \sum_{t\in I_i} (x_t^{\top}\Delta_i)^2 \leq 2 \sum_{i=1}^3 \sum_{t\in I_i} \varepsilon_t x_t^{\top} \Delta_i + 3C_3 \lambda^2 d_0 + 2\gamma \nonumber \\
	\leq & 2\sum_{i=1}^3 \left \| \frac{1}{\sqrt {|I_i| }} \sum_{t\in I_i } \varepsilon_t x_t  \right \|_{\infty } \| \sqrt { | I_i|}\Delta_i \|_1 + 3C_3 \lambda^2 d_0 + 2\gamma  \nonumber \\
   \le & \lambda/2 \sum_{i=1}^3 \left(  \sqrt {d_0 |I_i| }\|\Delta_i(S)\|_2  + \sqrt {|I_i| } \|\Delta_i(S^c) \|_1 \right) + 3C_3 \lambda^2 d_0 + 2\gamma,  \nonumber
	\end{align}
	where the last inequality follows from \Cref{lem-x-bound}.
	
It follows from identical arguments in \Cref{lemma:one change point} that, with probability at least $1 - 3c_1 (n \vee p)^{-288^2C_x^2d_0c_2/c_x^2} - 2(n \vee p)^{-c_3}$,
	\[
		\min\{|I_1|,\, |I_2|\} \le C_{\epsilon} \left(\frac{\lambda^2 d_0 + \gamma}{\kappa^2}\right).
	\]
	Since $|I_2| \ge \Delta$ by assumption, it follows from \Cref{assume:high dim coefficient}(\textbf{d}) that
	\[
		|I_1|  \le C_{\epsilon} \left(\frac{\lambda^2 d_0 + \gamma}{\kappa^2}\right),
	\]
	which contradicts \eqref{eq-contra-2}.

\vskip 3mm
\noindent {\bf Case 2.}  If  
	\[
		|I_3| \leq 288^2 C^2_x d_0 \log(n \vee p)/c_x^2,
	\]
	then it follows from \Cref{lemma:RSS lasso} that the following holds with probability at least $1 - 2c_1 (n \vee p)^{-288^2C_x^2d_0c_2/c_x^2} - 2(n \vee p)^{-c_3}$ that,
	\begin{align}
		& \sum_{t\in I } (y_t -x_t^{\top} \widehat \beta^\lambda_ I )^2 \leq \sum_{t\in I_1 } (y_t -x_t^{\top} \widehat \beta^\lambda_ {I_1}  )^2 + \sum_{t\in I_2 } (y_t -x_t^{\top} \widehat \beta^\lambda_ {I_2}  )^2  + \sum_{t\in I_3} (y_t -x_t^{\top} \widehat \beta^\lambda_ {I_3}  )^2 + 2\gamma \nonumber \\
		\leq &  \sum_{t \in I_1 } (y_t -x_t^{\top}   \beta^* _ {I_1}  )^2 + \sum_{t \in I_2 } (y_t -x_t^{\top}   \beta^* _ {I_2}  )^2 + \sum_{t \in I_3 } (y_t -x_t^{\top}   \beta^* _ {I_3}  )^2 + 2C_3 \lambda^2 d_0 + C_4 \lambda \sqrt{\log(p)} d_0^{3/2} + 2\gamma,  \nonumber
	\end{align}
	which implies that
	\begin{align}
		& \sum_{i=1}^3 \sum_{t\in I_i} (x_t^{\top}\Delta_i)^2 \leq 2 \sum_{i=1}^3 \sum_{t\in I_i} \varepsilon_t x_t^{\top} \Delta_i + 2C_3 \lambda^2 d_0 + C_4 \lambda \sqrt{\log(p)} d_0^{3/2} + 2\gamma \nonumber \\
		\leq & 2\sum_{i=1}^2 \left \| \frac{1}{\sqrt {|I_i| }} \sum_{t\in I_i } \varepsilon_t x_t  \right \|_{\infty } \| \sqrt { | I_i|}\Delta_i \|_1 + 2C_3 \lambda^2 d_0 + C_4 \lambda \sqrt{\log(p)} d_0^{3/2} \nonumber \\
		& \hspace{5cm} + 2\gamma + \sum_{t\in I_3} (x_t^{\top}\Delta_3)^2 + \sum_{t \in I_3}\varepsilon_t^2 \nonumber \\
	   \le & \lambda/2 \sum_{i=1}^2 \left(  \sqrt {d_0 |I_i| }\|\Delta_i(S)\|_2  + \sqrt {|I_i| } \|\Delta_i(S^c) \|_1 \right) + 2C_3 \lambda^2 d_0 + C_4 \lambda \sqrt{\log(p)} d_0^{3/2} \nonumber \\
	   & \hspace{5cm} + 2\gamma + \sum_{t\in I_3} (x_t^{\top}\Delta_3)^2 + \sum_{t \in I_3}\varepsilon_t^2. \nonumber 
	\end{align}
	The rest follows from the same arguments as in \textbf{Case 1.}
	
\end{proof}

\begin{lemma}[Case (iii) in Proposition~\ref{proposition:dp step 1}] \label{lemma:no change point}
For \Cref{assume:change point regression model}, under \Cref{assume:high dim coefficient}, if there exists no true change point in $I = (s, e]$, with 
	\[
		\lambda \geq \lambda_2 = C_{\lambda}\sigma_{\varepsilon} \sqrt{d_0 \log(n \vee p)}, 
	\]
	where $C_{\lambda} > \max\{8C_1C_x, \, 8C_{\beta}C_x/\sigma_{\varepsilon}\}$, and $\gamma = C_{\gamma} \sigma_{\varepsilon}^2 d_0^2\log(n \vee p)$, where $C_{\gamma} > \max\{3C_3/c_x^2, \, 3C_4/c_x\}$, it holds with probability at least $1 - 3c_1 (n \vee p)^{-288^2C_x^2d_0c_2/c_x^2} - 2(n \vee p)^{-c_3}$ that
	\[
		\sum_{t \in I}(y_t - x_t^{\top} \widehat \beta^\lambda_I)^2 < \min_{b = s + 1, \ldots, e - 1} \left\{\sum_{t \in (s, b]}(y_t - x_t^{\top} \widehat \beta^\lambda_{(s, b]})^2 + \sum_{t \in (b, e]} (y_t - x_t^{\top} \widehat \beta^\lambda_{(b, e]})^2 \right\} + \gamma.
	\]
\end{lemma}
\begin{proof}
First we notice that with the choice of $\lambda$, it holds that $\lambda > \lambda_1$, therefore we can apply \Cref{lemma:RSS lasso} when needed.

For any $b = s + 1, \ldots, e - 1$, let $I_1 = (s, b]$ and $I_2 = (b, e]$.  It follows from \Cref{lemma:RSS lasso} that with probability at least $1 - 3c_1 (n \vee p)^{-288^2C_x^2d_0c_2/c_x^2} - 2(n \vee p)^{-c_3}$, 
	\begin{align*}
		\max_{J \in \{I_1, I_2, I\}} \left | \sum_{t\in J} (y_t -x_t^{\top} \widehat \beta^\lambda_J)^2  - \sum_{t \in J} (y_t -x_t^{\top}   \beta^*_J)^2 \right| \leq \max\left\{C_3\lambda^2 d_0, \, C_4 \lambda \sqrt{\log(n \vee p)} d_0^{3/2} \right\} < \gamma/3.
	\end{align*}
	Since $ \beta^*_ {I} = \beta^*_ {I_1} = \beta^*_ {I_2}$, the final claim holds automatically.

\end{proof}

\begin{lemma}[Case (iv) in Proposition~\ref{proposition:dp step 1}] \label{lemma:three change point}
For \Cref{assume:change point regression model}, under \Cref{assume:high dim coefficient}, if $I = (s, e]$ contains  $J$ true change points $\{ \eta_{k}\}_{j=1}^J$, where $|J| \ge 3$, if
	\[
		\lambda \geq \lambda_2 = C_{\lambda}\sigma_{\varepsilon} \sqrt{d_0} \log(n \vee p), 
	\]
	where $C_{\lambda} > 8C_{\beta}C_x/\sigma_{\varepsilon}$, then with probability at least $1 - nc_1 (n \vee p)^{-288^2C_x^2d_0c_2/c_x^2} - 2(n \vee p)^{-c_3}$,
	\[
		\sum_{t\in I } (y_t -x_t^{\top} \widehat \beta^\lambda_ I )^2 > \sum_{j=1}^{J+1} \sum_{t\in I_j  } (y_t -x_t^{\top} \widehat \beta^\lambda_ {I_j}  )^2 + J \gamma, 
	\]
	where $I_1 =(s,\eta_{1}] $, $I_j = (\eta_{j},\eta_{j+1}] $ for any $2\le j \le J $ and $I_{J+1} = (\eta_J, e]$.
\end{lemma} 

\begin{proof}
First we notice that with the choice of $\lambda$, it holds that 
	\[
		\lambda \geq \max\{\lambda_1, \, \lambda_2\},
	\]
	and therefore we can apply Lemmas~\ref{lemma:lasso}, \ref{lemma:RSS lasso} and \ref{lemma:oracle 2} when needed.

We prove the claim by contradiction, assuming that
	\[
		\sum_{t\in I } (y_t -x_t^{\top} \widehat \beta^\lambda_ I )^2 \leq \sum_{j=1}^{J+1} \sum_{t\in I_j  } (y_t -x_t^{\top} \widehat \beta^\lambda_ {I_j}  )^2 + J \gamma.
	\]
	Let $\Delta_i = \widehat \beta^\lambda_I  - \beta_{I_i}^*$, $i = 1, \ldots, J+1$.  It then follows from \Cref{lemma:RSS lasso} that with probability at least $1 - nc_1 (n \vee p)^{-288^2C_x^2d_0c_2/c_x^2} - 2(n \vee p)^{-c_3}$,
	\begin{align}
		& \sum_{t \in I} (y_t - x_t^{\top} \widehat \beta^\lambda_I )^2 \leq \sum_{j=1}^{J+1} \sum_{t \in I_j} (y_t - x_t^{\top} \widehat \beta^\lambda_{I_j})^2 + J \gamma \nonumber \\
		& \hspace{2cm}\leq  \sum_{j=1}^{J+1} \sum_{t \in I_j} (y_t - x_t^{\top} \beta^*_{I_j})^2 + J \gamma + (J+1) C_{\gamma} \sigma_{\varepsilon}^2 d_0^2\log(n \vee p), \nonumber
	\end{align}
	which implies that 
	\begin{align} 
		\sum_{j=1}^{J+1} \sum_{t\in I_j} (x_t^{\top}\Delta_j)^2 \leq 2 \sum_{j=1}^{J+1}  \sum_{t\in I_j} \varepsilon_t x_t^{\top} \Delta_j + J \gamma + (J+1)C_{\gamma} \sigma_{\varepsilon}^2 d_0^2\log(n \vee p). \label{eq:three change point first step}
	\end{align}

\vskip 3mm
\noindent \textbf{Step 1.}  For any $j \in \{2, \ldots, J\}$, it follows from \Cref{assume:high dim coefficient} that
	\begin{equation} \label{eq-lem12-pf-ijlength}
		|I_j| \geq \Delta \geq 288^2 C^2_x d_0 \log(n \vee p)/c_x^2.
	\end{equation}
	Due to \Cref{lem-x-bound}, on the event $\mathcal{B}_{(0, n]}$, it holds that	
	\begin{align}  
		& \sum_{t\in I_j} \varepsilon_t x_t^{\top} \Delta_j \leq \left\|\frac{1}{\sqrt{|I_j|}}  \sum_{t\in I_j } \varepsilon_t x_t \right\|_{\infty} \| \sqrt { | I_j| } \Delta_j\|_1  \leq \lambda/4 \left(\sqrt{d_0 |I_j|}\|\Delta_j(S)\|_2 + \sqrt{|I_j|} \|\Delta_j(S^c)\|_1\right) \nonumber  \\
		\leq & \frac{4\lambda^2 d_0}{c_x^2} + \frac{c_x^2|I_j|}{256} \|\Delta_j\|_2^2  + \lambda/4 \sqrt {|I_j|} \|\widehat \beta^\lambda_{I}(S^c)\|_1 \nonumber \\
		= & \frac{4\lambda^2 d_0}{c_x^2} + \frac{c_x^2|I_j|}{256} \|\Delta_j\|_2^2  + \lambda/4 \sqrt {|I_j|} \|(\widehat \beta^\lambda_{I} - (|I|)^{-1}\sum_{t\in I}\beta^*_t)(S^c)\|_1 \nonumber \\
		\leq & \frac{4\lambda^2 d_0}{c_x^2} + \frac{c_x^2|I_j|}{256} \|\Delta_j\|_2^2  + \lambda/4 \sqrt {|I_j|} \|\widehat \beta^\lambda_{I} - (|I|)^{-1}\sum_{t\in I}\beta^*_t\|_1 \nonumber \\
		\leq & \frac{4\lambda^2 d_0}{c_x^2} + \frac{c_x^2|I_j|}{256} \|\Delta_j\|_2^2 + C_5/4 \lambda^2 d_0, \label{eq:three change points standard equality 1}
	\end{align}
	where the third inequality follows from $2ab \leq a^2 + b^2$, letting
	\[
		a = 2\lambda\sqrt{d_0}/c_x \quad \mbox{and} ]\quad b = c_x\sqrt{|I_j|}\|\Delta_j\|_2/16,
	\] 
	and the last inequality follows from \Cref{lemma:oracle 2}.  In addition, on the event of $\mathcal{E}_{I_j}$, due to \Cref{lem-RE}, it holds that 
	\begin{align}
		& \sqrt{\sum_{t\in I_j} (x_t^{\top} \Delta_j)^2} \geq \frac{c_x\sqrt{|I_j|}}{4} \|\Delta_j\|_2 - 9C_x \sqrt{\log(p)} \|\Delta_j\|_1 \nonumber \\
		\geq & \frac{c_x\sqrt{|I_j|}}{4} \|\Delta_j\|_2 - 9C_x \sqrt{d_0 \log(p)} \|\Delta_j\|_2 - 9C_x \sqrt{\log(p)} \|\Delta_j(S^c)\|_1 \nonumber \\
		\geq & \frac{c_x\sqrt{|I_j|}}{8} \|\Delta_j\|_2 - 9C_x \sqrt{\log(p)} \|\Delta_j(S^c)\|_1 \geq \frac{c_x \sqrt{|I_j|}}{8} \|\Delta_j\|_2 - \frac{9 C \lambda d_0 \sqrt{\log(p)}}{\sqrt{|I|}}, \label{eq:REC three change}
	\end{align}
	where the third inequality follows from \eqref{eq-lem12-pf-ijlength} and the last follows from \Cref{lemma:oracle 2}.

\vskip 3mm
\noindent \textbf{Step 2.}  We then discuss the intervals $I_1$ and $I_{J+1}$.  These two will be treated in the same way, and therefore for $L \in \{I_1, I_{J+1}\}$ and $l \in \{1, J+1\}$, we have the following arguments.  If $|L| \geq 288^2 C^2_x d_0 \log(n \vee p)/c_x^2$, then due to the same arguments in \textbf{Step 1}, \eqref{eq:three change points standard equality 1} and \eqref{eq:REC three change} hold.  If instead, $|L| < 288^2 C^2_x d_0 \log(n \vee p)/c_x^2$ holds, then
	\begin{align*}  
		\sum_{t\in L} \varepsilon_t x_t^{\top} \Delta_l \leq 2^{-1}\sum_{t \in L}(x_t^{\top}\Delta_l)^2 + 4\sum_{t \in L}\varepsilon_t^2.
	\end{align*}
	Therefore, it follows from \eqref{eq:three change point first step} that
	\[
		\sum_{j = 2}^J |I_j|c_x^2 \|\Delta_j\|_2^2  \leq JC\max\left\{\lambda^2 d_0, \, \lambda \sqrt{\log(n \vee p)} d_0^{3/2}\right\} + J\gamma.
	\]

\vskip 3mm
\noindent \textbf{Step 3.} Since for any $j \in \{2, \ldots, J-1\}$, it holds that
	\begin{align*} 
		|I_j|\|\Delta_j\|_2^2 + |I_{j+1}|\|\Delta_{j+1}\|_2^2 & \geq   \inf_{v \in \mathbb{R}^p}\bigl\{|I_j|\|\beta_{I_j}^* - v\|_2^2 + |I_{j+1}| \|\beta_{I_{j+1}}^* -v\|_2^2\bigr\} \\
		& \geq \frac{|I_j||I_{j+1}|}{|I_j|+|I_{j+1}|} \kappa^2 \geq \min\{|I_j|,\, |I_{j+1}|\}\kappa^2/2.
	\end{align*}
	It then follows from the same arguments in \Cref{lemma:one change point} that
	\[
		\min_{j = 2, \ldots, J-1} |I_j| \leq C_{\epsilon} \left(\frac{\lambda^2 d_0 + \gamma}{\kappa^2}\right),
	\]
	which is a contradiction to \eqref{eq-lem12-pf-ijlength}.
\end{proof}

\subsection{Proof of Proposition~\ref{prop-2}}

\begin{lemma} \label{lem-case-5-prop-2-needed}
Under the assumptions and notation in \Cref{proposition:dp step 1}, suppose there exists no true change point in the interval $I$.  For any interval $J \supset I$, with 
	\[
		\lambda \geq \lambda_2 = C_{\lambda}\sigma_{\varepsilon} \sqrt{d_0 \log(n \vee p)}, 
	\]
	where $C_{\lambda} > \max\{8C_1C_x, \, 8C_{\beta}C_x/\sigma_{\varepsilon}\}$, it holds that with probability at least $1 - c_1 (n \vee p)^{-288^2C_x^2d_0c_2/c_x^2} - 2(n \vee p)^{-c_3}$,
	\[
		\sum_{t \in I} (y_t - x_t^{\top} \beta^*_I)^2 - \sum_{t \in I} (y_t - x_t^{\top} \widehat{\beta}^{\lambda}_J)^2 \leq C_6\lambda^2 d_0.
	\]
\end{lemma}
 
\begin{proof}

\noindent \textbf{Case 1.} If 
	\begin{equation}\label{eq-lem16-i-cond} 
		|I| \geq  288^2 C^2_x d_0 \log(n \vee p)/c_x^2,		
	\end{equation}
	then letting $\Delta_I = \beta^*_I  -\widehat \beta^{\lambda}_J$, on the event $\mathcal{E}_I$, we have
	\begin{align}
		& \sqrt{\sum_{t \in I} (x_t^{\top} \Delta_I)^2} \geq \frac{c_x\sqrt{|I|}}{4} \|\Delta_I\|_2 - 9C_x\sqrt{\log(p)} \|\Delta_I\|_1 \nonumber \\
		= & \frac{c_x\sqrt{|I|}}{4} \|\Delta_I\|_2 - 9C_x\sqrt{\log(p)} \|\Delta_I(S)\|_1 - 9C_x\sqrt{\log(p)} \|\Delta_I(S^c)\|_1 \nonumber \\
		\geq & \frac{c_x\sqrt{|I|}}{4} \|\Delta_I\|_2 - 9C_x\sqrt{d_0\log(p)} \|\Delta_I\|_2 - 9C_x\sqrt{\log(p)} \|\Delta_I(S^c)\|_1 \nonumber \\
		\geq & \frac{c_x\sqrt{|I|}}{8} \|\Delta_I\|_2 - 9C_x\sqrt{\log(p)} \|\widehat{\beta}^{\lambda}_J(S^c)\|_1 \geq \frac{c_x\sqrt{|I|}}{8} |\Delta_I\|_2 - 9C_5 C_x d_0 \lambda\log^{1/2}(p), \label{eq-lem16-pf-1}
	\end{align} 
	where the last inequality follows from \Cref{lemma:oracle 2}.  We then have on the event $\mathcal{B}_I$,
	\begin{align*}
	& \sum_{t \in I} (y_t - x_t^{\top} \beta^*_I)^2 - \sum_{t \in I} (y_t - x_t^{\top} \widehat{\beta}^{\lambda}_J)^2  = 2 \sum_{t \in I} \varepsilon_t x_t^{\top}\Delta_I  - \sum_{t \in I} (x_t^{\top}\Delta_I)^2 \\
	\leq & 2\left \|\sum_{t \in I }  x_t \varepsilon_t \right\|_{\infty }  \left( \sqrt {d_0} \| \Delta_I(S) \|_2  + \|\widehat \beta^{\lambda}_J (S^c) \|_1  \right) \\
	& \hspace{3cm} - \frac{c_x^2 |I|}{64}\|\Delta_I\|_2^2 - \frac{81C_5^2 C_x^2 \lambda^2 d_0^2\log(p)}{c_x^4 |I|} + \frac{9C_5C_xd_0\lambda \log^{1/2}(p) \|\Delta_I\|_2}{4}\\  
	\leq & \frac{\lambda}{2}\sqrt{d_0}\|\Delta_I\|_2 + \frac{\lambda^2 d_0 C_5}{2c_x^2 \sqrt{|I|}} - \frac{c_x^2 |I|}{64}\|\Delta_I\|_2^2 + \frac{9C_5C_xd_0\lambda \log^{1/2}(p) \|\Delta_I\|_2}{4} \\
	\leq & \frac{\lambda}{2}\sqrt{d_0}\|\Delta_I\|_2 + \frac{\lambda^2 \sqrt{d_0} C_5}{576c_x\sqrt{\log(n \vee p)}C_x} - 36^2C_x^2d_0\log(n \vee p)\|\Delta_I\|^2_2 + \frac{9C_5C_xd_0\lambda \log^{1/2}(p) \|\Delta_I\|_2}{4} \\
	\leq & \frac{\lambda^2}{16C_x^2} + d_0C_x^2 \|\Delta_I\|^2_2 + \frac{\lambda^2 \sqrt{d_0} C_5}{576c_x\sqrt{\log(n \vee p)}C_x} - 36^2C_x^2d_0\log(n \vee p)\|\Delta_I\|^2_2 \\
	& \hspace{3cm} + d_0\log(p)\|\Delta_I\|_2^2C_x^2 + \frac{81C_5^2d_0\lambda^2}{64} \\
	\leq & C_6\lambda^2 d_0.
	\end{align*}
	where the first inequality follows from \eqref{eq-lem16-pf-1}, the second inequality follows from event $\mathcal{B}_I$ and \Cref{lemma:oracle 2}, the third follows from the \eqref{eq-lem16-i-cond}, the fourth follows from $2ab \leq a^2 + b^2$, first letting
	\[
		a = \lambda/(4C_x) \quad \mbox{and} \quad b = \sqrt{d_0}C_x \|\Delta_I\|_2,
	\]
	then letting
	\[
		a = C_x\sqrt{d_0 \log(p)}\|\Delta_I\|_2 \quad \mbox{and} \quad b = 9C_5\sqrt{d_0}\lambda/8,
	\]
	and the last inequality follows from \Cref{lemma:oracle 2}. 

\vskip 3mm	
\noindent \textbf{Case 2.} If $ |I| \leq 288^2 C^2_x d_0 \log(n \vee p)/c_x^2$, then with probability at least $1 - 2(n \vee p)^{-c}$, 
	\begin{align*}
		& \sum_{t \in I} (y_t - x_t^{\top} \beta^*_I)^2 - \sum_{t \in I} (y_t - x_t^{\top} \widehat \beta^{\lambda}_J)^2 = 2 \sum_{t \in I} \varepsilon_t x_t^{\top}(\widehat{\beta}^{\lambda}_J - \beta^*_I) - \sum_{t \in I} \{x_t^{\top}(\beta^*_I  -\widehat\beta^{\lambda}_J)\}^2 \\
		\leq & \sum_{t \in I} \varepsilon_t^2 \leq \max\{\sqrt{|I|\log(n \vee p)}, \, \log(n \vee p)\} \leq C_6\lambda^2 d_0.
	\end{align*}
\end{proof}

\begin{proof}[Proof of \Cref{prop-2}] 
Denote $S^*_n = \sum_{t=1}^n (y_t - x_t^{\top} \beta^*_t)^2$.  Given any collection $\{t_1, \ldots, t_m\}$, where $t_1 < \cdots < t_m$, and $t_0 = 0$, $t_{m+1} = n$, let 
	\begin{equation}\label{eq-sn-def}
		S_n(t_1, \ldots, t_{m}) = \sum_{k=1}^{m} \sum_{t = t_k +1}^{t_{k+1}} \bigl(y_t -  x_t^{\top}\widehat \beta_{(t_{k}, t_{k+1}]}^\lambda\bigr)^2. 
	\end{equation}
	For any collection of time points, when defining \eqref{eq-sn-def}, the time points are sorted in an increasing order.
	
Let $\{ \widehat \eta_{k}\}_{k=1}^{\widehat K}$ denote the change points induced by $\widehat {\mathcal P}$.  If one can justify that 
	\begin{align}
		S^*_n + K\gamma  \ge  &S_n(\eta_1,\ldots,\eta_K)   + K\gamma - C_3(K+1) d_0 \lambda^2 \label{eq:K consistency step 1} \\ 
		\ge & S_n (\widehat \eta_{1},\ldots, \widehat \eta_{\widehat K } ) +\widehat K \gamma - C_3(K+1) d_0 \lambda^2 \label{eq:K consistency step 2} \\ 
		\ge &   S_n ( \widehat \eta_{1},\ldots, \widehat \eta_{\widehat K } , \eta_1,\ldots,\eta_K ) + \widehat K \gamma  - 2C(K+1)d_0 \lambda^2 - C_3(K+1) d_0 \lambda^2 \label{eq:K consistency step 3}
	\end{align}
	and that 
	\begin{align}\label{eq:K consistency step 4}
		S^*_n  -S_n ( \widehat \eta_{1},\ldots, \widehat \eta_{\widehat K } , \eta_1,\ldots,\eta_K ) \le C (K + \widehat{K} + 2) \lambda^2 d_0,
	\end{align}
	then it must hold that $| \hatp | = K$, as otherwise if $\widehat K \ge K+1 $, then  
	\begin{align*}
		C (K + \widehat{K} + 2) \lambda^2 d_0 & \geq S^*_n  -S_n ( \widehat \eta_{1},\ldots, \widehat \eta_{\widehat K } , \eta_1,\ldots,\eta_K ) \\
		& \geq -3C (K+1) \lambda^2 d_0 + (\widehat K - K)\gamma \geq C_{\gamma} (K+1)\lambda^2 d_0.
	\end{align*} 
	Therefore due to the assumption that $| \hatp|  =\widehat K\le 3K $, it holds that 
	\begin{align} \label{eq:Khat=K} 
		C(5K + 3)\lambda^2d_0  \geq (\widehat K - K)\gamma \geq \gamma,
	\end{align}
	Note that \eqref{eq:Khat=K} contradicts the choice of $\gamma$.

Note that \eqref{eq:K consistency step 1}  is implied by 
	\begin{align}\label{eq:step 1 K consistency}  
		\left| 	S^*_n   -   S_n(\eta_1,\ldots,\eta_K)    \right| \le  C_3(K+1) d_0 \lambda^2,
	\end{align}
	which is  immediate consequence  of \Cref{lemma:RSS lasso}.  Since $\{ \widehat \eta_{k}\}_{k=1}^{\widehat K}$ are the change points induced by $\widehat {\mathcal P}$, \eqref{eq:K consistency step 2} holds because $\hatp$ is a minimiser.

For every $\widehat I =(s,e]\in \hatp$ denote
	\[
		\widehat I  =  (s ,\eta_{p+1}]\cup \ldots \cup  (\eta_{p+q} ,e]  = J_1 \cup \ldots  \cup J_{q+1},
	\]
	where $\{ \eta_{p+l}\}_{l=1}^{q+1}  =\widehat I \ \cap \ \{\eta_k\}_{k=1}^K$. 
	Then \eqref{eq:K consistency step 3} is an immediate consequence of the following inequality
	\begin{align} \label{eq:one change point step 3}
		\sum_{t\in \widehat I }(y_t -x_t^{\top} \widehat \beta^\lambda_{\widehat I} )^2  \ge \sum_{l=1}^{q+1}  \sum_{t \in J_l}(y_t -x_t^{\top}  \widehat \beta^\lambda_{J_l})^2 - C(q+1) \lambda^2 d_0.
	\end{align}
	By \Cref{lemma:RSS lasso}, it holds that
	\begin{align} 
		\sum_{l=1}^{q+1}  \sum_{t \in J_l}(y_t -x_t^{\top}  \widehat \beta^\lambda_{J_l})^2 & \le \sum_{l=1}^{q+1} \sum_{t \in J_l}(y_t -x_t^{\top}    \beta^*_t  )^2  + (q+1) \max \left\{C_3 d_0\lambda^2, \, C_4\lambda \sqrt{\log(n \vee p)} d_0^{3/2} \right\} \nonumber \\
		& = \sum_{t\in \widehat I }(y_t -x_t^{\top}    \beta^*_t  )^2  + (q+1) \max \left\{C_3 d_0\lambda^2, \, C_4\lambda \sqrt{\log(n \vee p)} d_0^{3/2} \right\}.  \label{eq:K consistency step 3 inequality}
	\end{align}
	Then for each $l \in \{1, \ldots, q+1\}$,
	\[
		\sum_{t\in J_l }(y_t -x_t^{\top}  \widehat \beta^\lambda_{\widehat I}   )^2   \ge\sum_{t\in J_l}(y_t -x_t^{\top}\beta^*_t  )^2 - C_6\lambda^2 d_0,
	\]
	where  the  inequality follows from  \Cref{lem-case-5-prop-2-needed}.  Therefore the above inequality implies that 
	\begin{align} \label{eq:K consistency step 3 inequality  3}  \sum_{t\in \widehat I }(y_t -x_t^{\top}  \widehat \beta^\lambda_{\widehat I}   )^2   \ge\sum_{t\in \widehat I }(y_t -x_t^{\top}    \beta^*_t  )^2   -  C_6(q+1)\lambda^2 d_0.
	\end{align}
	Note that \eqref{eq:K consistency step 3 inequality} and  \eqref{eq:K consistency step 3 inequality 3} implies \eqref{eq:one change point step 3}.

Finally, to show \eqref{eq:K consistency step 4}, observe that from \eqref{eq:step 1 K consistency}, it suffices to show that 
	\[
		S_n(\eta_1,\ldots,\eta_K)  -  S_n ( \widehat \eta_{1},\ldots, \widehat \eta_{\widehat K } , \eta_1,\ldots,\eta_K )  \le C (K+\widehat K) \lambda^2,
	\]
	the analysis of which follows from a similar but simpler argument as above.
\end{proof} 
 
\section{Proof of Corollary~\ref{cor-lr-3}}\label{sec-pf-cor-3}

\begin{lemma}
Let $\mathcal{S}$ be any linear subspace in $\mathbb{R}^n$ and $\mathcal{N}_{1/4}$	be a $1/4$-net of $\mathcal{S} \cap B(0, 1)$, where $B(0, 1)$ is the unit ball in $\mathbb{R}^n$.  For any $u \in \mathbb{R}^n$, it holds that
	\[
		\sup_{v \in \mathcal{S} \cap B(0, 1)} \langle v, u \rangle \leq 2 \sup_{v \in \mathcal{N}_{1/4}} \langle v, u \rangle,
	\]
	where $\langle \cdot, \cdot \rangle$ denotes the inner product in $\mathbb{R}^n$.
\end{lemma}

\begin{proof}
Due to the definition of $\mathcal{N}_{1/4}$, it holds that for any $v \in \mathcal{S} \cap B(0, 1)$, there exists a $v_k \in \mathcal{N}_{1/4}$, such that $\|v - v_k\|_2 < 1/4$.  Therefore,
	\begin{align*}
		\langle v, u \rangle = \langle v - v_k + v_k, u \rangle = \langle x_k, u \rangle + \langle v_k, u \rangle \leq \frac{1}{4} \langle v, u \rangle + \frac{1}{4} \langle v^{\perp}, u \rangle + \langle v_k, u \rangle,
	\end{align*}
	where the inequality follows from  $x_k = v - v_k = \langle x_k, v \rangle v + \langle x_k, v^{\perp} \rangle v^{\perp}$.  Then we have
	\[
		\frac{3}{4}\langle v, u \rangle \leq \frac{1}{4} \langle v^{\perp}, u \rangle + \langle v_k, u \rangle.
	\]
	It follows from the same argument that 
	\[
		\frac{3}{4}\langle v^{\perp}, u \rangle \leq \frac{1}{4} \langle v, u \rangle + \langle v_l, u \rangle,
	\]
	where $v_l \in \mathcal{N}_{1/4}$ satisfies $\|v^{\perp} - v_l\|_2 < 1/4$.  Combining the previous two equation displays yields
	\[
		\langle v, u \rangle \leq 2 \sup_{v \in \mathcal{N}_{1/4}} \langle v, u \rangle,
	\]
	and the final claims holds.
\end{proof}

\Cref{lem-wang-lem-3} is an adaptation of Lemma~3 in \cite{wang2019statistically}.

\begin{lemma}\label{lem-wang-lem-3}
	For data generated from \Cref{assume:change point regression model}, for any interval $I = (s, e] \subset \{1, \ldots, n\}$, it holds that for any $\delta > 0$, $i \in \{1, \ldots, p\}$, 
	\[
		\mathbb{P}\left\{\sup_{\substack{v \in \mathbb{R}^{(e-s)}, \, \|v\|_2 = 1\\ \sum_{t = 1}^{e-s-1} \mathbbm{1}\{v_i \neq v_{i+1}\} = m}} \left|\sum_{t = s+1}^e v_t \varepsilon_t x_t(i)\right| > \delta \right\} \leq C(e-s-1)^m 9^{m+1} \exp \left\{-c \min\left\{\frac{\delta^2}{4C_x^2}, \, \frac{\delta}{2C_x \|v\|_{\infty}}\right\}\right\}.
	\]
\end{lemma}

\begin{proof}
For any $v \in \mathbb{R}^{(e-s)}$ satisfying $\sum_{t = 1}^{e-s-1}\mathbbm{1}\{v_i \neq v_{i+1}\} = m$, it is determined by a vector in $\mathbb{R}^{m+1}$ and a choice of $m$ out of $(e-s-1)$ points.  Therefore we have,
	\begin{align*}
		& \mathbb{P}\left\{\sup_{\substack{v \in \mathbb{R}^{(e-s)}, \, \|v\|_2 = 1\\ \sum_{t = 1}^{e-s-1} \mathbbm{1}\{v_i \neq v_{i+1}\} = m}} \left|\sum_{t = s+1}^e v_t \varepsilon_t x_t(i)\right| > \delta \right\} \\
		\leq & {(e-s-1) \choose m} 9^{m+1} \sup_{v \in \mathcal{N}_{1/4}}	\mathbb{P}\left\{\left|\sum_{t = s+1}^e v_t \varepsilon_t x_t(i)\right| > \delta/2 \right\} \\
		\leq & {(e-s-1) \choose m} 9^{m+1} C\exp \left\{-c \min\left\{\frac{\delta^2}{4C_x^2}, \, \frac{\delta}{2C_x \|v\|_{\infty}}\right\}\right\} \\
		\leq & C(e-s-1)^m 9^{m+1} \exp \left\{-c \min\left\{\frac{\delta^2}{4C_x^2}, \, \frac{\delta}{2C_x \|v\|_{\infty}}\right\}\right\}.
	\end{align*}

\end{proof}

\begin{proof}[Proof of \Cref{cor-lr-3}]
For each $k \in \{1, \ldots, K\}$, let
	\[
		\widehat{\beta}_t = \begin{cases}
 				\widehat{\beta}_1, & t \in \{s_k + 1, \ldots, \widehat{\eta}_k\}, \\
 				\widehat{\beta}_2, & t \in \{\widehat{\eta}_k + 1, \ldots, e_k\}.
 			\end{cases}		
	\]
	
Without loss of generality, we assume that $s_k < \eta_k < \widehat{\eta}_k < e_k$.  We proceed the proof discussing two cases.

\vskip 3mm
\noindent \textbf{Case (i).}  If 
	\[
		\widehat{\eta}_k - \eta_k < \max\{288^2 C^2_x d_0 \log(n \vee p)/c_x^2, \, C_{\varepsilon}\log(n \vee p)/\kappa^2\},
	\]	
	then the result holds.

\vskip 3mm
\noindent \textbf{Case (ii).}  If 
	\begin{equation}\label{eq-pf-cor-6-int-con}
		\widehat{\eta}_k - \eta_k \geq \max\{288^2 C^2_x d_0 \log(n \vee p)/c_x^2, \, C_{\varepsilon}\log(n \vee p)/\kappa^2\},
	\end{equation} 
	then we first to prove that with probability at least $1 - C(n \vee p)^{-c}$,
	\[
		\sum_{t = s_k + 1}^{e_k}\|\widehat{\beta}_t - \beta^*_t\|_2^2 \leq C_1d_0 \zeta^2 = \delta.
	\]	

Due to \eqref{eq-g-lasso}, it holds that
	\begin{align}\label{eq-lem19-pf-1}
		& \sum_{t = s_k + 1}^{e_k} \|y_t - x_t^{\top}\widehat{\beta}_t\|_2^2 + \zeta \sum_{i = 1}^p \sqrt{\sum_{t = s_k + 1}^{e_k} \bigl(\widehat{\beta}_t\bigr)_i^2} \leq \sum_{t = s_k + 1}^{e_k} \|y_{t} - x_t^{\top}\beta^*_t \|_2^2 + \zeta \sum_{i = 1}^p \sqrt{\sum_{t = s_k + 1}^{e_k} \bigl(\beta^*_t\bigr)_{i}^2}.
	\end{align}
	Let $\Delta_t = \widehat{\beta}_t - \beta^*_t$.  It holds that
	\[
		\sum_{t = s_k + 1}^{e_k - 1}\mathbbm{1}\left\{\Delta_t \neq \Delta_{t+1}\right\} = 2.
	\]
	Eq.\eqref{eq-lem19-pf-1} implies that
	\begin{align}\label{eq-lem19-pf-2}
		\sum_{t = s_k + 1}^{e_k} \|\Delta_t^{\top} x_t\|_2^2 + \zeta  \sum_{i = 1}^p \sqrt{\sum_{t = s_k + 1}^{e_k} \bigl(\widehat{\beta}_t\bigr)_{i}^2}\leq 2\sum_{t = s_k + 1}^{e_k} (y_t - x_t^{\top}\beta^*_t)\Delta_t^{\top} x_t + \zeta \sum_{i = 1}^p \sqrt{\sum_{t = s_k + 1}^{e_k} \bigl(\beta^*_t\bigr)_{i}^2}.
	\end{align}

Note that
	\begin{align}
		& \sum_{i = 1}^p \sqrt{\sum_{t = s_k + 1}^{e_k} \bigl(\beta^*_t\bigr)_{i}^2} - \sum_{i = 1}^p \sqrt{\sum_{t = s_k + 1}^{e_k} \bigl(\widehat{\beta}_t\bigr)_{i}^2} = \sum_{i \in S} \sqrt{\sum_{t = s_k + 1}^{e_k} \bigl(\beta^*_t\bigr)_{i}^2} - \sum_{i \in S} \sqrt{\sum_{t = s_k + 1}^{e_k} \bigl(\widehat{\beta}_t\bigr)_{i}^2}  - \sum_{i \in S^c} \sqrt{\sum_{t = s_k + 1}^{e_k} \bigl(\widehat{\beta}_t\bigr)_{i}^2} \nonumber \\
		\leq & \sum_{i \in S} \sqrt{\sum_{t = s_k + 1}^{e_k} \bigl(\Delta_t\bigr)_{i}^2}  - \sum_{i \in S^c} \sqrt{\sum_{t = s_k + 1}^{e_k} \bigl(\Delta_t\bigr)_{i}^2}.\label{eq-lem19-pf-3}
	\end{align}

We then examine the cross term, with probability at least $1 - C(n \vee p)^{-c}$, which satisfies the following
	\begin{align}
		& \left|\sum_{t = s_k + 1}^{e_k} (y_t - x_t^{\top}\beta^*_t)\Delta_t^{\top} x_t\right| = \left|\sum_{t = s_k + 1}^{e_k} \varepsilon_{t} \Delta_t^{\top} x_t\right| = \sum_{i = 1}^p \left\{\left|\frac{\sum_{t = s_k + 1}^{e_k} \varepsilon_{t} \Delta_t(i) x_t(i)}{\sqrt{\sum_{t = s_k + 1}^{e_k} \left(\Delta_t(i)\right)^2}}\right| \sqrt{\sum_{t = s_k + 1}^{e_k} \left(\Delta_t(i)\right)^2}\right\} \nonumber \\
		\leq & \sup_{i = 1, \ldots, p} \left|\frac{\sum_{t = s_k + 1}^{e_k} \varepsilon_{t} \Delta_t(i) X_t(i)}{\sqrt{\sum_{t = s_k + 1}^{e_k} \left(\Delta_t(i)\right)^2}}\right| \sum_{i = 1}^p \sqrt{\sum_{t = s_k + 1}^{e_k} \left(\Delta_t(i)\right)^2} \leq (\zeta/4) \sum_{i = 1}^p \sqrt{\sum_{t = s_k + 1}^{e_k} \left(\Delta_t(i)\right)^2}, \label{eq-lem19-pf-4}
	\end{align}
	where the second inequality follows from \Cref{lem-wang-lem-3} and \eqref{eq-pf-cor-6-int-con}.

Combining \eqref{eq-lem19-pf-1}, \eqref{eq-lem19-pf-2}, \eqref{eq-lem19-pf-3} and \eqref{eq-lem19-pf-4} yields  
	\begin{equation}\label{eq-lem19-pf-5}
		\sum_{t = s_k + 1}^{e_k} \|\Delta_t^{\top} x_t\|_2^2	 + \frac{\zeta}{2}\sum_{i \in S^c} \sqrt{\sum_{t = s_k + 1}^{e_k} \bigl(\Delta_t\bigr)_{i}^2} \leq \frac{3\zeta}{2}\sum_{i \in S} \sqrt{\sum_{t = s_k + 1}^{e_k} \bigl(\Delta_t\bigr)_{i}^2}.
	\end{equation}
	
Now we are to explore the restricted eigenvalue inequality.  Let
	\[
		I_1 = (s_k, \eta_k], \quad I_2 = (\eta_k, \widehat{\eta}_k], \quad I_3 = (\widehat{\eta}_k, e_k].
	\]
	We have that with probability at least $1 - C(n \vee p)^{-c}$, on the event $\cap_{i = 1, 3}\mathcal{E}_{I_i}$,
	\begin{align*}
		& \sum_{t = s_k + 1}^{e_k} \|\Delta_t^{\top} x_t\|_2^2 = \sum_{i = 1}^3 \sum_{t \in I_i}\|\Delta_{I_i}^{\top}x_t\|_2^2 \geq \sum_{i = 1, 3}\sum_{t \in I_i}\|\Delta_{I_i}^{\top}x_t\|_2^2 \\
		\geq & \sum_{i = 1, 3} \left(\frac{c_x\sqrt{|I_i|}}{4} \|\Delta_{I_i}\|_2 - 9C_x \sqrt{\log(p)} \|\Delta_{I_i}\|_1\right)^2 \\
		\geq & \sum_{i = 1, 3} \left(\frac{c_x\sqrt{|I_i|}}{8} \|\Delta_{I_i}\|_2 - 9C_x \sqrt{\log(p)} \|\Delta_{I_i}(S^c)\|_1\right)^2,
	\end{align*}
	where the last inequality follows from \eqref{eq-lr-cond-1} and \Cref{assume:high dim coefficient}, that 
	\[
		\min\{|I_1|, \, |I_3|\} > (1/3)\Delta > 288^2 C^2_x d_0 \log(n \vee p)/c_x^2.
	\]
	Since $|I_2| > 288^2 C^2_x d_0 \log(n \vee p)/c_x^2$, we have
	\[
		\sqrt{\sum_{t \in I_2}\|\Delta_{I_2}^{\top}x_t\|_2^2} \geq \frac{c_x\sqrt{|I_2|}}{8} \|\Delta_{I_2}\|_2 - 9C_x \sqrt{\log(p)} \|\Delta_{I_2}(S^c)\|_1.
	\]

Note that
	\begin{align*}
		& \sqrt{\sum_{i = 1}^3 \left(\sum_{j \in S^c} |\Delta_{I_i}(j)|\right)^2 } \leq \sqrt{\sum_{i = 1}^3\left(\sqrt{\frac{|I_i|}{I_0}} \sum_{j \in S^c}|\Delta_{I_i}(j)| \right)^2}	 \\
		\leq & \sum_{j \in S^c} I_0^{-1/2} \sqrt{\sum_{t = s_k + 1}^{e_k} (\Delta_t(i))^2} \leq 3\sum_{j \in S} I_0^{-1/2} \sqrt{\sum_{t = s_k + 1}^{e_k} (\Delta_t(i))^2} \\
		\leq & I_0^{-1/2}3 \sqrt{d_0 \sum_{j \in S} \sum_{t = s_k + 1}^{e_k} (\Delta_t(i))^2} \leq \frac{c_x}{96 C_x \sqrt{\log(n \vee p)}} \sqrt{\sum_{t = s_k + 1}^{e_k} \|\Delta_t\|_2^2}.
	\end{align*}

Therefore,
	\begin{align*}
		& \frac{c_x}{8}\sqrt{\sum_{t = s_k + 1}^{e_k} \|\Delta_t\|_2^2} - \frac{3c_x}{32 C_x \sqrt{\log(n \vee p)}} \sqrt{\sum_{t = s_k + 1}^{e_k} \|\Delta_t\|_2^2} \\
		\leq & \sum_{i = 1}^3 \frac{c_x\sqrt{|I_i|}}{8} \|\Delta_{I_i}\|_2 - \frac{3c_x}{32 C_x \sqrt{\log(n \vee p)}} \sqrt{\sum_{t = s_k + 1}^{e_k} \|\Delta_t\|_2^2} \leq \sqrt{3} \sqrt{ \sum_{t = s_k + 1}^{e_k} \|\Delta_t^{\top} x_t\|_2^2 } \\
		\leq & \frac{3\sqrt{\zeta}}{\sqrt{2}}d_0^{1/4} \left(\sum_{t = s_k + 1}^{e_k} \|\Delta_t\|_2^2\right)^{1/4} \leq \frac{18\zeta d_0^{1/2}}{c_x} + \frac{c_x}{16}\sqrt{\sum_{t = s_k + 1}^{e_k} \|\Delta_t\|_2^2}
	\end{align*}
	where the last inequality follows from \eqref{eq-lem19-pf-5} and which implies
	\[
		\frac{c_x}{32}\sqrt{\sum_{t = s_k + 1}^{e_k} \|\Delta_t\|_2^2} \leq \frac{18\zeta d_0^{1/2}}{c_x}
	\]	
	Therefore,
	\[
		\sum_{t = s_k + 1}^{e_k}\|\widehat{\beta}_t - \beta^*_t\|_2^2 \leq 576^2\zeta^2d_0/c_x^4.
	\]

\vskip 3mm

  Let $\beta^*_1 = \beta^*_{\eta_k}$ and $\beta^*_2 = \beta^*_{\eta_k + 1}$.  We have that
	\[
		\sum_{t = s_k + 1}^{e_k}\|\widehat{\beta}_t - \beta^*_t\|_2^2 = |I_1| \|\beta^*_1 - \widehat{\beta}_1\|_2^2 + |I_2| \|\beta^*_2 - \widehat{\beta}_1\|_2^2 + |I_3| \|\beta^*_2 - \widehat{\beta}_2\|_2^2.
	\]
	
Since
	\begin{align*}
		& \eta_k - s_k = \eta_k - \frac{2}{3}\widetilde{\eta}_k - \frac{1}{3}\widetilde{\eta}_k \\
		= & \frac{2}{3}(\eta_k - \eta_{k-1}) + \frac{2}{3}(\widetilde{\eta}_k - \eta_k) -  \frac{2}{3}(\widetilde{\eta}_{k-1} - \eta_{k-1}) + (\eta_k - \widetilde{\eta}_k) \\
		\geq & \frac{2}{3}\Delta - \frac{1}{3}\Delta = \frac{1}{3}\Delta,
	\end{align*}	
	where the inequality follows from \Cref{assume:high dim coefficient} and \eqref{eq-lr-cond-1}, we have that
	\[
		\Delta\|\beta^*_1 - \widehat{\beta}_1\|_2^2/3 \leq |I_1| \|\beta^*_1 - \widehat{\beta}_1\|_2^2 \leq \delta \leq \frac{C_1 C_{\zeta}^2\Delta \kappa^2} {C_{\mathrm{SNR}}d_0 K \sigma^2_{\epsilon} \log^{\xi}(n \vee p) } \leq c_1 \Delta \kappa^2,
	\]
	where $1/4 > c_1 > 0$ is an arbitrarily small positive constant.  Therefore we have
	\[
		\|\beta^*_1 - \widehat{\beta}_1\|_2^2 \leq c_1 \kappa^2.
	\]
	
In addition we have
	\[
		\|\beta^*_2 - \widehat{\beta}_1\|_2 \geq \|\beta^*_2 - \beta^*_1\|_2 - \|\beta^*_1 - \widehat{\beta}_1\|_2 \geq \kappa/2.
	\]	
	Therefore, it holds that 	
	\[
		\kappa^2 |I_2|/4 \leq |I_2| \|\beta^*_2 - \widehat{\beta}_1\|_2^2 \leq \delta,
	\]
	which implies that 
	\[
		|\widehat{\eta}_k - \eta_k| \leq \frac{4C_1d_0 \zeta^2}{\kappa^2}.
	\]
\end{proof}

\section{Lower bounds}

\begin{proof}[Proof of \Cref{lem-reg-lb-1}]
For any vector $\beta$, if $x \sim \mathcal{N}(0, I_p)$, $\epsilon \sim \mathcal{N}(0, \sigma^2)$ and $y = x^{\top} \beta + \epsilon$, then we denote
	\begin{align*}
		\left(\begin{array}{c} 
			y \\
			x
		\end{array}\right) \sim \mathcal{N}\left(0, \Sigma_\beta\right), \quad \mbox{where}  \quad \Sigma_{\beta} = \left(\begin{array}{cc} 
			\beta^{\top}\beta + \sigma^2 & \beta^{\top} \\
			\beta & I
 		\end{array}\right).
	\end{align*}

Now for a fixed $S \subset \{1, \ldots, p\}$ satisfying $|S| = d$, define
	\[
		\mathcal{S} = \left\{u \in \mathbb{R}^p: \, u_i = 0, i \notin S; u_i = \kappa/\sqrt{d} \mbox{ or } -\kappa/\sqrt{d}, i \in S\right\}.
	\]
	Define 
	\[
		P_0 = \mathcal{N}\left(0, \Sigma_0\right) \quad \mbox{and} \quad P_u = \mathcal{N}\left(0, \Sigma_u\right), \quad \forall u \in \mathcal{S},
	\]
	where 
	\[
		\Sigma_0 = \left(\begin{array}{cc}
			\sigma^2 & 0 \\
			0 & I_p
		 \end{array}\right) \quad \mbox{and} \quad \Sigma_u = \left(\begin{array}{cc}
			\sigma^2 + \kappa^2 & u^{\top} \\
			u & I_p	
		 \end{array}\right).
	\]

\medskip
\noindent \textbf{Step 1.}  Let $P_{0, u}^T$ denote the joint distribution of independent random vectors $\{Z_i = (y_i, x_i^{\top})^{\top}\}_{i = 1}^T \subset \mathbb{R}^{p+1}$ such that
	\[
		Z_1, \ldots, Z_{\Delta} \stackrel{\mbox{iid}}{\sim} \mathcal{N}(0, \Sigma_u) \quad \mbox{and} \quad Z_{\Delta + 1}, \ldots, Z_T \stackrel{\mbox{iid}}{\sim} \mathcal{N}(0, \Sigma_0).
	\] 
	Let $P_{1, u}^T$ denote the joint distribution of independent random vectors $\{Z_i = (y_i, x_i^{\top})^{\top}\}_{i = 1}^T \subset \mathbb{R}^{p+1}$ such that
	\[
		Z_1, \ldots, Z_{T-\Delta} \stackrel{\mbox{iid}}{\sim} \mathcal{N}(0, \Sigma_0) \quad \mbox{and} \quad Z_{T-\Delta + 1}, \ldots, Z_T \stackrel{\mbox{iid}}{\sim} \mathcal{N}(0, \Sigma_u).
	\] 
	For $i \in \{0, 1\}$, let
	\[
		P_i = 2^{-d} \sum_{u \in \mathcal{S}} P_{i, u}^T.
	\]	
	Let $\eta(P)$ denote the change point location of a distribution $P$.  Then since $\eta(P_{0, u}) = \Delta$ and $\eta(P_{1, u}) = T-\Delta$ for any $u \in \mathcal{S}$, we have that
	\[
		|\eta(P_0) - \eta(P_1)| = T - 2\Delta \geq T/2,
	\]
	due to the fact that $\Delta \leq T/4$.  It follows from Le Cam's lemma \citep{yu1997assouad} that
	\[
		\inf_{\widehat{\eta}} \sup_{P \in \mathcal{P}} \mathbb{E}_P(|\widehat{\eta} - \eta|) \geq T/2 (1 - d_{\mathrm{TV}}(P_0, P_1)),
	\]
	where $d_{\mbox{TV}}(P_0, P_1) = \|P_0 - P_1\|_1/2$, with $\|P_0 - P_1\|_1$ denoting the $L_1$ distance between the Lebesgue densities of the distributions $P_0$ and $P_1$.  Then we have that
	\[
		\inf_{\widehat{\eta}} \sup_{P \in \mathcal{P}} \mathbb{E}_P(|\widehat{\eta} - \eta|) \geq T/2 (1 - 2^{-1}\|P_0 - P_1\|_1).
	\]

\medskip
\noindent \textbf{Step 2.}  Let $P_0^{\Delta}$ be the joint distribution of 
	\[
		Z_1, \ldots, Z_{\Delta} \stackrel{\mbox{iid}}{\sim} \mathcal{N}(0, \Sigma_0)
	\]	
	and $P_1^{\Delta} = 2^{-d} \sum_{u \in \mathcal{S}} P_{1, u}^{\Delta}$, where $P^{\Delta}_1$ is the joint distribution of 
	\[
		Z_1, \ldots, Z_{\Delta} \stackrel{\mbox{iid}}{\sim} \mathcal{N}(0, \Sigma_u).  
	\]
	It follows from Step 2 in the proof of Lemma 3.1 in \cite{wang2017optimal} that 
	\[
		\|P_0 - P_1\|_1 \leq 2 \|P_0^{\Delta} - P_1^{\Delta}\|_1,
	\]
	which leads to 
	\[
		\inf_{\widehat{\eta}} \sup_{P \in \mathcal{P}} \mathbb{E}_P(|\widehat{\eta} - \eta|) \geq T/2 (1 - \|P_0^{\Delta} - P_1^{\Delta}\|_1) \geq T/2 (1 - \sqrt{\chi^2(P_1^{\Delta}, P_0^{\Delta})}),
	\]
	where the last inequality follows from \cite{tsybakov2009introduction}.
	
Note that
	\begin{align*}
		\chi^2(P_1^{\Delta}, P_0^{\Delta}) & = \mathbb{E}_{P_p^{\Delta}}\left\{\left(\frac{dP_1^{\Delta}}{dP_0^{\Delta}} - 1\right)^2\right\} = \frac{1}{4^d} \sum_{u, v \in \mathcal{S}} \mathbb{E}_{P_0^{\Delta}} \left(\frac{dP^{\Delta}_u dP^{\Delta}_v}{dP^{\Delta}_0 dP^{\Delta}_0}\right) - 1 \\
		& = \frac{1}{4^d} \sum_{u, v \in \mathcal{S}} \left\{\mathbb{E}_{P_0} \left(\frac{dP_u dP_v}{dP_0 dP_0}\right)\right\}^{\Delta} - 1.
	\end{align*}
\medskip
\noindent \textbf{Step 3.}  For any $u, v \in \mathcal{S}$, we have that
	\begin{align*}
		& \mathbb{E}_{P_0} \left(\frac{dP_u dP_v}{dP_0 dP_0}\right) \\
		= & \frac{|\Sigma_u|^{-1/2} |\Sigma_v|^{-1/2}}{|\Sigma_0|^{-1/2}} (2\pi)^{-\frac{p+1}{2}} \int_{\mathbb{R}^{p+1}} \exp \left\{-\frac{z^{\top}(\Sigma_u^{-1} + \Sigma_v^{-1} - \Sigma_0^{-1})z}{2}\right\}\, dz \\
		= & \frac{|\Sigma_u|^{-1/2} |\Sigma_v|^{-1/2}}{|\Sigma_0|^{-1/2}} |\Sigma_u^{-1} + \Sigma_v^{-1} - \Sigma_0^{-1}|^{-1/2}.
	\end{align*}
	
In addition, we have that
	\[
		|\Sigma_u| = |\Sigma_v| = |\Sigma_0| = \sigma^2, \quad \Sigma^{-1}_0 = \left(\begin{array}{cc} 
			\sigma^{-2} & 0 \\
			0 & I 
		\end{array}\right),
	\]	
	\[
		\Sigma^{-1}_u = \left(\begin{array}{cc} 
			\sigma^{-2} & -\sigma^{-2}u^{\top} \\
			-\sigma^{-2} u & I + \sigma^{-2}u u^{\top}
		\end{array}\right) \quad \mbox{and}  \quad \Sigma^{-1}_v = \left(\begin{array}{cc} 
			\sigma^{-2} & -\sigma^{-2}v^{\top} \\
			-\sigma^{-2} v & I + \sigma^{-2} v v^{\top}
		\end{array}\right).
	\]
	Then
	\[
		\mathbb{E}_{P_0} \left(\frac{dP_u dP_v}{dP_0 dP_0}\right) = \sigma^p \left|\left(\begin{array}{cc}
			1 & -(u+v)^{\top} \\
			-(u+v) & \sigma^2 I_p + uu^{\top} + vv^{\top}				
		\end{array}\right)\right|^{-1/2} = \sigma^p |M|^{-1/2}.
	\]
	
Note that
	\[
		|M| = \left|\left\{1 - (u+v)^{\top} \left(\sigma^2 I_p + uu^{\top} + vv^{\top}\right)^{-1}(u+v)\right\}\right| \left|\sigma^2 I_p + uu^{\top} + vv^{\top}\right|.
	\]	
	As for the matrix $M_1 = \sigma^2 I_p + uu^{\top} + vv^{\top}$, since $u, v \neq 0$, there are two cases.  Let 
	\[
		\rho_{u, v} = \frac{u^{\top} v}{\kappa^2}.
	\]
	\begin{itemize}
	\item The dimension of the linear space spanned by $u$ and $v$ is one, i.e.~$|\rho| = 1$.  In this case, for any $w \perp \mathrm{span}\{u\}$, $\|w\|_2 = 1$, it holds that
		\[
			M_1 w = \sigma^2 w.
		\] 	
		There are $p-1$ such linearly independent $w$.  For any $w \in \mathrm{span}\{u\}$, $\|w\|_2 = 1$, it holds that
		\[
			M_1 w = (\sigma^2 + 2\kappa^2) w.
		\]
		Then $|M_1| = \sigma^{2p-2}(\sigma^2 + 2\kappa^2)$.
		
		If $\rho_{u, v} = -1$, then $|M| = |M_1| = \sigma^{2p-2}(\sigma^2 + 2\kappa^2)$.
		
		If $\rho_{u, v} = 1$, then 
		\begin{align*}
			|M| & = \left|1 - 4u^{\top} \frac{u}{\kappa} \frac{1}{\sigma^2 + 2\kappa^2} \frac{u^{\top}}{\kappa}u\right|\sigma^{2p-2}(\sigma^2 + 2\kappa^2) \\
			& = \frac{|\sigma^2 -  2\kappa^2|}{\sigma^2 + 2\kappa^2} \sigma^{2p-2}(\sigma^2 + 2\kappa^2) = \sigma^{2p-2}|\sigma^2 - 2\kappa^2|.
		\end{align*}
		
		Therefore in this case
		\[
			|M| = \sigma^{2p-2}|\sigma^2 - 2\rho_{u, v}\kappa^2|.
		\]

	\item The dimension of the linear space spanned by $u$ and $v$ is two, i.e.~$|\rho| < 1$.  In this case, for any $w \perp \mathrm{span}\{u\}$, $\|w\|_2 = 1$, it holds that
		\[
			M_1 w = \sigma^2 w.
		\] 	
		There are $p-2$ such linearly independent $w$.  
		
		We also have
		\[
			M_1 \frac{u+v}{\|u+v\|} = (\sigma^2 + \kappa^2 + \rho_{u, v}\kappa^2)\frac{u+v}{\|u+v\|}
		\]
		and
		\[
			M_1 \frac{u-v}{\|u-v\|} = (\sigma^2 + \kappa^2 - \rho_{u, v}\kappa^2)\frac{u-v}{\|u-v\|}.
		\]
		Then
		\[
			|M_1| = \sigma^{2p-4}(\sigma^2 + \kappa^2 + \rho_{u, v}\kappa^2)(\sigma^2 + \kappa^2 - \rho_{u, v}\kappa^2)
		\]
		
		In addition, 
		\begin{align*}
			& (u+v)^{\top} \left(\sigma^2 I_p + uu^{\top} + vv^{\top}\right)^{-1}(u+v) \\
			= & (u+v)^{\top}	\frac{u+v}{\|u+v\|} \frac{1}{\sigma^2 + \kappa^2 + \rho_{u, v}\kappa^2}\left(\frac{u+v}{\|u+v\|}\right)^{\top} (u+v) \\
			&\hspace{1cm} + (u+v)^{\top}	\frac{u-v}{\|u-v\|} \frac{1}{\sigma^2 + \kappa^2 - \rho_{u, v}\kappa^2}\left(\frac{u-v}{\|u-v\|}\right)^{\top} (u+v)  \\
			= & \frac{2\kappa^2 + 2\kappa^2 \rho_{u, v}}{\sigma^2 + \kappa^2 + \rho_{u, v}\kappa^2}.
		\end{align*}
		Then,
		\[
			|M| = \sigma^{2p-4} |\sigma^2 - \kappa^2 - \rho_{u, v}\kappa^2| (\sigma^2 + \kappa^2 - \rho_{u, v}\kappa^2),
		\]
		which is consistent with the case when $|\rho_{u, v}| = 1$.
	\end{itemize}

We then have
	\[
		\mathbb{E}_{P_0} \left(\frac{dP_u dP_v}{dP_0 dP_0}\right) = \left|1 - \frac{\kappa^2}{\sigma^2} - \frac{u^{\top}v}{\sigma^2}\right|^{-1/2} \left|1 + \frac{\kappa^2}{\sigma^2} - \frac{u^{\top}v}{\sigma^2}\right|^{-1/2}.
	\]
	
Due to the fact that $cd/\Delta < 1/4$, we have that
	\begin{align*}
		1 - \frac{\kappa^2}{\sigma^2} - \frac{u^{\top}v}{\sigma^2} \geq 1 - \frac{2\kappa^2}{\sigma^2} \geq 1 - \frac{2cd}{\Delta} > 0,
	\end{align*}
	then
	\begin{align*}
		& \mathbb{E}_{P_0} \left(\frac{dP_u dP_v}{dP_0 dP_0}\right) = \left(1 - \frac{\kappa^2}{\sigma^2} - \frac{u^{\top}v}{\sigma^2}\right)^{-1/2} \left(1 + \frac{\kappa^2}{\sigma^2} - \frac{u^{\top}v}{\sigma^2}\right)^{-1/2} \\
		= & \left(1 - \frac{2u^{\top}v}{\sigma^2} - \frac{\kappa^4}{\sigma^4} + \frac{(u^{\top}v)^2}{\sigma^4}\right)^{-1/2} \leq \left(1 - \frac{2u^{\top}v}{\sigma^2} - \frac{\kappa^4}{\sigma^4}\right)^{-1/2}
	\end{align*}
	
Then we have
	\begin{align*}
		& \chi^2(P_1^{\Delta}, P_0^{\Delta}) \leq \frac{1}{4^d} \sum_{u, v \in \mathcal{S}} \left(1 - \frac{2u^{\top}v}{\sigma^2} - \frac{\kappa^4}{\sigma^4}\right)^{-\Delta/2} - 1 \\
		= & \mathbb{E}_{U, V} \left\{1 - \frac{\kappa^2}{\sigma^2} (U^{\top}V/d)^2 - \frac{\kappa^4}{\sigma^4}\right\}^{-\Delta/2} - 1 = \mathbb{E}_{V} \left\{1 - \frac{\kappa^2}{\sigma^2} (1^{\top}V/d)^2 - \frac{\kappa^4}{\sigma^4}\right\}^{-\Delta/2} - 1 \\
		\leq & \mathbb{E} \left\{\exp\left(\frac{\kappa^2 \Delta}{\sigma^2} \varepsilon_d + \frac{\kappa^4 \Delta}{\sigma^4}\right)\right\} - 1,
	\end{align*}
	where $U$ and $V$ are two independent $d$-dimensional Radamacher random vectors, $\varepsilon_d = (1^{\top}V/d)^2$, and the last inequality follows from $(1-t)^{-\Delta/2} \leq \exp(\Delta t)$, for any $t \leq 1/2$.
	
Due to the Hoeffding inequality, it holds that for any $\lambda > 0$,
	\[
		\mathbb{P}(\varepsilon_d \geq \lambda) \leq 2e^{-2d\lambda}.
	\]	
	Then
	\begin{align*}
		& \mathbb{E} \left\{\exp\left(\frac{\kappa^2 \Delta}{\sigma^2} \varepsilon_d + \frac{\kappa^4 \Delta}{\sigma^4}\right)\right\}	= \int_0^{\infty} \mathbb{P}\left\{\exp\left(\frac{\kappa^2 \Delta}{\sigma^2} \varepsilon_d + \frac{\kappa^4 \Delta}{\sigma^4}\right) \geq u\right\} \, du \\
		\leq & 1 + \int_1^{\infty} \mathbb{P}\left\{\frac{\kappa^2 \Delta}{\sigma^2} \varepsilon_d + \frac{\kappa^4 \Delta}{\sigma^4} \geq \log(u)\right\} \, du = 1 + \int_1^{\infty} \mathbb{P}\left\{\varepsilon_d \geq \frac{\log(u) - \frac{\kappa^4 \Delta}{\sigma^4}}{\frac{\kappa^2 \Delta}{\sigma^2}}\right\} \, du \\
		\leq & 1 + 2\int_1^{\infty} \exp\left\{-\frac{2d\sigma^2}{\kappa^2 \Delta} \log(u) + \frac{2d\kappa^2}{\sigma^2}\right\} \, du \\
		= & 1 + \frac{2\exp(2d\kappa^2\sigma^{-2})}{\frac{2d\sigma^2}{\kappa^2 \Delta} - 1} \leq 1 + \frac{2e}{2/c-1} \leq 5/4,
	\end{align*}
	where the last two inequalities hold due to
	\[
		2cd^2 \leq \Delta \quad \mbox{and} \quad	 c < \frac{2}{8e+1}.
	\]

We then complete the proof.
	
\end{proof}

\begin{proof}[Proof of \Cref{lem-reg-lb-2}]
For any vector $\beta$, if $x \sim \mathcal{N}(0, I_p)$, $\epsilon \sim \mathcal{N}(0, \sigma^2)$ and $y = x^{\top} \beta + \epsilon$, then we denote
	\begin{align*}
		\left(\begin{array}{c} 
			y \\
			x
		\end{array}\right) \sim \mathcal{N}\left(0, \Sigma_\beta\right), \quad \mbox{where}  \quad \Sigma_{\beta} = \left(\begin{array}{cc} 
			\beta^{\top}\beta + \sigma^2 & \beta^{\top} \\
			\beta & I
 		\end{array}\right).
	\end{align*}

Now for a fixed $S \subset \{1, \ldots, p\}$ satisfying $|S| = d$, define
	\[
		\mathcal{S} = \left\{u \in \mathbb{R}^p: \, u_i = 0, i \notin S; u_i = \kappa/\sqrt{d} \mbox{ or } -\kappa/\sqrt{d}, i \in S\right\}.
	\]
	Define 
	\[
		P_0 = \mathcal{N}\left(0, \Sigma_0\right) \quad \mbox{and} \quad P_u = \mathcal{N}\left(0, \Sigma_u\right), \quad \forall u \in \mathcal{S},
	\]
	where 
	\[
		\Sigma_0 = \left(\begin{array}{cc}
			\sigma^2 & 0 \\
			0 & I_p
		 \end{array}\right) \quad \mbox{and} \quad \Sigma_u = \left(\begin{array}{cc}
			\sigma^2 + \kappa^2 & u^{\top} \\
			u & I_p	
		 \end{array}\right).
	\]

\medskip
\noindent \textbf{Step 1.}  Let $P_{0, u}^T$ denote the joint distribution of independent random vectors $\{Z_i = (y_i, x_i^{\top})^{\top}\}_{i = 1}^T \subset \mathbb{R}^{p+1}$ such that
	\[
		Z_1, \ldots, Z_{\Delta} \stackrel{\mbox{iid}}{\sim} \mathcal{N}(0, \Sigma_u) \quad \mbox{and} \quad Z_{\Delta + 1}, \ldots, Z_T \stackrel{\mbox{iid}}{\sim} \mathcal{N}(0, \Sigma_0).
	\] 
	Let $P_{1, u}^T$ denote the joint distribution of independent random vectors $\{Z_i = (y_i, x_i^{\top})^{\top}\}_{i = 1}^T \subset \mathbb{R}^{p+1}$ such that
	\[
		Z_1, \ldots, Z_{\Delta + \delta} \stackrel{\mbox{iid}}{\sim} \mathcal{N}(0, \Sigma_u) \quad \mbox{and} \quad Z_{\Delta + \delta + 1}, \ldots, Z_T \stackrel{\mbox{iid}}{\sim} \mathcal{N}(0, \Sigma_0).
	\] 
	For $i \in \{0, 1\}$, let
	\[
		P_i = 2^{-d} \sum_{u \in \mathcal{S}} P_{i, u}^T.
	\]	
	Then we have that
	\[
		\inf_{\widehat{\eta}} \sup_{P \in \mathcal{P}} \mathbb{E}_P(|\widehat{\eta} - \eta|) \geq \delta (1 - 2^{-1}\|P_0 - P_1\|_1).
	\]

\medskip
\noindent \textbf{Step 2.}  Let $P_0^{\delta}$ be the joint distribution of 
	\[
		Z_1, \ldots, Z_{\delta} \stackrel{\mbox{iid}}{\sim} \mathcal{N}(0, \Sigma_0)
	\]	
	and $P_1^{\delta} = 2^{-d} \sum_{u \in \mathcal{S}} P_{1, u}^{\delta}$, where $P^{\delta}_{1, u}$ is the joint distribution of 
	\[
		Z_1, \ldots, Z_{\delta} \stackrel{\mbox{iid}}{\sim} \mathcal{N}(0, \Sigma_u).  
	\]

It follows from the identical arguments in the proof of \Cref{lem-reg-lb-1} that
	\[
		\inf_{\widehat{\eta}} \sup_{P \in \mathcal{P}} \mathbb{E}_P(|\widehat{\eta} - \eta|) \geq \delta (1 - \|P_0^{\delta} - P_1^{\delta}\|_1) \geq \delta (1 - \sqrt{\chi^2(P_1^{\delta}, P_0^{\delta})})
	\]
	and
	\begin{align*}
		\chi^2(P_1^{\delta}, P_0^{\delta}) = \frac{1}{4^d} \sum_{u, v \in \mathcal{S}} \left\{\mathbb{E}_{P_0} \left(\frac{dP_u dP_v}{dP_0 dP_0}\right)\right\}^{\delta} - 1 \leq \frac{2\exp(2d\kappa^2\sigma^{-2})}{\frac{2d\sigma^2}{\kappa^2 \delta} - 1}.
	\end{align*}

\medskip 
\noindent \textbf{Step 3.}  Let 
	\[
		\delta = \frac{Cd\sigma^2}{\kappa^2}. 
	\]
	We have that
	\[
		\chi^2(P_1^{\delta}, P_0^{\delta}) = 1/4,
	\]
	provided that $d^2\zeta_T\Delta^{-1} < 1$ and with $C = 2/(8e+1)$.  Then we conclude the proof.
\end{proof}

\bibliographystyle{authordate1}
\bibliography{citations}

\end{document}